\documentclass{jocg}
\pdfoutput=1
\usepackage{amsmath,amsfonts}

\title{\MakeUppercase{Optimally Fast Incremental Manhattan Plane Embedding and Planar Tight Span Construction}\thanks{This work was supported in part by NSF grant 0830403 and by the Office of Naval Research under grant N00014-08-1-1015.}}
\author{David Eppstein\thanks{\affil{Computer Science Department, University of California, Irvine},
    \email{eppstein@uci.edu}}}

\usepackage{graphicx}
\usepackage{cite}
\usepackage{hyperref}

\def\R{\ensuremath{\Bbb R}}

\usepackage{amsthm}
\theoremstyle{plain}
\newtheorem{theorem}{Theorem}
\newtheorem{lemma}{Lemma}
\newtheorem{corollary}{Corollary}

\begin{document}
\maketitle   

\begin{abstract}
We describe a data structure, a \emph{rectangular complex}, that can be used to represent hyperconvex metric spaces that have the same topology (although not necessarily the same distance function) as subsets of the plane.
We show how to use this data structure to construct the tight span of a metric space given as an $n\times n$ distance matrix, when the tight span is homeomorphic to a subset of the plane, in time $O(n^2)$, and to add a single point to a planar tight span in time $O(n)$. As an application of this construction, we show how to test whether a given finite metric space embeds isometrically into the Manhattan plane in time $O(n^2)$, and add a single point to the space and re-test whether it has such an embedding in time $O(n)$.
\end{abstract}

\section{Introduction}
Any metric space $(X,d)$, finite or infinite, can be embedded without distortion into a larger metric space $(T_X,d)$ called variously its \emph{tight span}, \emph{injective envelope}, \emph{hyperconvex hull}, or \emph{convex hull}~\cite{ChrLar-Algs-94,Dev-AoC-06,Dre-AiM-84,DreHubMou-DM-01,DreHubMou-AC-01,DreHubMou-AM-02,HerJos-CDM-07,Isb-CMH-64,StuYu-EJC-04}.  The tight span is an invariant of the given metric space~-- isometric metric spaces have isometric tight spans -- and it has properties similar to those of $\ell_\infty$ spaces. In particular, tight spans are \emph{hyperconvex}, meaning that their metric balls satisfy the Helly property: any system of balls that intersect pairwise has a common intersection. Additionally, they are \emph{injective}, meaning that whenever $X$ can be embedded without distortion into a hyperconvex space, the embedding extends to~$T_X$. By means of the tight span, algorithms and mathematical results designed for the simpler case of hyperconvex metric spaces can often be extended to any other kind of metric space.

As well as being a fundamental construction in the theory of metric spaces themselves, tight spans have close connections to several other problems of algorithmic interest:
\begin{itemize}
\item
The $\ell_1$-metric plane, and $\ell_\infty$ spaces of any dimension, are hyperconvex. Therefore, by injectivity, a given metric space embeds isometrically into one of these spaces if and only if its tight span does. As we show, this allows the tight span to be used for efficient tests of two-dimensional $\ell_1$- or $\ell_\infty$-embeddability of finite metric spaces, a problem previously studied by several authors~\cite{AviDez-Nw-91,BanChe-Nw-98,Edm-DCG-08,MalMal-DCG-92}.
\item
Tight spans generalize the concept of an \emph{orthogonal convex hull}, in the sense that, when the input space $X$ is a subset of the $\ell_1$ plane, with its distance function $d$ being the Manhattan distance, then the tight span of $(X,d)$ is essentially the same as the orthogonal convex hull of $X$. (For a more precise statement of this equivalence, see Lemma~\ref{lem:L1-span}.)
\item
In biology, tight spans may be used to represent evolutionary distances between species in as tree-like a way as possible: if $(X,d)$ is the set of distances within a tree, then the tight span is a topological complex of points and line segments in the structure of a tree, with the same pairwise distances. For distances that are a small perturbation from a tree metric, the tight span provides a space that still retains tree-like properties but that does not arbitrarily single out any specific tree among the family of all trees consistent with the data~\cite{DreHubMou-DM-01,DreHubMou-AM-02}; indeed, any \emph{Steiner tree} whose path lengths dominate the distances in the input metric space but are otherwise as short as possible may be embedded isometrically into the tight span.
\item
The \emph{$k$-server problem} is a standard model of online decision-making in which an algorithm is presented with an online sequence requests from points in a metric space, and must field the requests by moving one of the servers to the requested point, while minimizing the total distance that the servers move. Although a $k$-server algorithm needs only to move a single server discretely to the requested point, it is a convenient conceptual device to view all servers as moving simultaneously and continuously towards the request point until one of them reaches that point and serves the request. Adopting this point of view, and using the tight span as the metric space within which the continuous motion of the servers takes place, has led to advances in online $k$-server algorithms~\cite{ChrLar-SJDM-91,ChrLar-Algs-94}.
\item
In facility location, the points of the tight span represent distance vectors from a single facility to each of the locations in a given input metric space, such that the distance between any two input locations is less than the sum of their distances to the facility, and that no other distance vector with the same property has smaller distances to each location. Problems of optimal facility location such as the problem of determining a star-topology network that minimizes the sum of the distances to the facility or that minimizes the dilation of the path between any pair of locations may therefore be formulated as a search for the optimal point in the tight span~\cite{EppWor-WADS-09}.
\item
Several recent algorithmic works in computational topology have involved two simplicial complexes defined from a set of points and a distance threshold (radius), the {\v C}ech complex and the Vietoris-Rips complex~\cite{ChaEriWor-SoCG-08,Ghr-BAMS-08,GhrMuh-IPSN-05,Zom-CG-10}. The {\v C}ech complex has a simplex for every set of balls of the given radius that have a common intersection, whereas the Vietoris-Rips complex has a simplex for every set of balls that pairwise intersect; because this pairwise intersection property is based only on distances between points, and not on the geometry of the space they embed into, it is easier to compute and may be applied even when only distances and not coordinates are known, as may happen in a sensor network. The tight span unifies these two concepts by providing a geometric space into which a given distance matrix can be embedded, within which the {\v C}ech complex and the Vietoris-Rips complex coincide: by hyperconvexity, any pairwise-intersecting set of balls has a common intersection.
\end{itemize}

Because of these applications, it is of interest to design algorithms that can take as input a metric space with finitely many points (described by an $n\times n$ distance matrix) and produce as output a concise representation of its tight span. And indeed, the tight span of any finite metric space may be represented in a combinatorial way, and constructed algorithmically, as the set of bounded faces of an $n$-dimensional intersection of halfspaces~\cite{Dre-AiM-84,StuYu-EJC-04}. However, for distance matrices satisfying a general position assumption, the largest dimension of these faces is high, between $n/3$ and $n/2$~\cite{Dev-AoC-06}, and the tight span has exponential worst-case complexity~\cite{HerJos-CDM-07}. Additionally, although there has been some research on methods for computing the bounded faces of a halfspace intersections without wasting time on the unbounded faces~\cite{EppLof-SoCG-11,HerJosPfe-10}, it is not known how to solve this problem with an output-sensitive polynomial time bound given only the halfspaces as input. Therefore, the applications of this polyhedral representation of tight spans are limited to spaces with very small numbers of points. It would be of interest to find special cases of finite metric spaces for which the tight span can be represented and constructed more efficiently, avoiding the limitations of the polyhedral approach.

As we show in this work, one such special case arises when the tight span is \emph{planar} (topologically equivalent to a subset of the plane). Although the metrics with planar tight spans cannot be in general position (for $n>6$), they include a wide class of examples including the distance metrics of unweighted \emph{squaregraphs}, the planar graphs in which every bounded face is a quadrilateral and every vertex that is not on the unbounded face has degree at least four~\cite{BanCheEpp-SJDM-10,Epp-TA-09}. As we show, for metrics that have a planar tight span, the tight span can be represented and constructed much more efficiently than the bounds for the general-purpose halfspace intersection algorithm would suggest.

\subsection{New results}
We show the following results:
\begin{itemize}
\item We show how to represent planar tight spans as \emph{rectangular complexes} formed by gluing together points, line segments, and $\ell_1$ rectangles. Our representation is suitable for use as a data structure in algorithms involving tight spans, and can be augmented to allow fast distance queries.
\item We describe an algorithm that, given an $n$-point metric space with a planar tight span, constructs a complex representing its tight span in $O(n^2)$ time. Our algorithm is \emph{incremental} in the sense that a single additional point can be added to the metric space, and the tight span augmented to include it, in $O(n)$ time.
\item As an illustration of the power of tight spans for algorithmic problems on metric spaces, we show how to embed any finite metric space into the Manhattan plane\footnote{The $\ell_1$ and $\ell_\infty$ planes may be mapped isometrically into each other by a transform that scales the plane by a $\sqrt 2$ factor and rotates it by a $45^\circ$ degree angle; when the distinction between these two metrics is unimportant we call them interchangeably the \emph{Manhattan plane}.} without distortion (if such an embedding exists) in $O(n^2)$ time. The main idea is to find a rectangular complex representing the planar tight span of the input metric space and then, once the tight span has been constructed, test whether the rectangular complex representing it obeys a set of simple local conditions on its structure. As we show, these local conditions are necessary and sufficient to allow the complex, and therefore also the input metric space, to be embedded into the Manhattan plane.
As with our tight span algorithm, this algorithm is incremental: any single point can be added to the metric space, and the embedding updated, in $O(n)$ time.
\end{itemize}

Our algorithms are optimal: their quadratic running time matches the size of the input distance matrix.
In order to avoid issues of numerical precision we assume a model of computation in which distances may be computed as exact numbers, and in which simple arithmetic formulas and comparisons involving distances may be computed exactly. This model is reasonable for the classes of inputs known to have planar tight spans, such as unweighted graph distance metrics in certain graphs. This model would not be appropriate, however, for situations in which the distances come from physical measurements. It would be of interest to determine, given an inexact distance matrix, whether there exists a planar hyperconvex metric space that closely approximates but does not exactly match its distances, generalizing the work of B{\u a}diou on approximate Manhattan plane embedding~\cite{Bad-SODA-03}, but we leave this problem of constructing approximate planar tight spans for future work.

\subsection{Related work}
In previous work~\cite{Epp-TA-09} we proved the hyperconvexity of a class of metric spaces that are homeomorphic to subsets of the plane, and we showed that several interesting graph-based metrics have tight spans that belong to this class. In particular this is true for the metric of unweighted shortest path distances in squaregraphs, as well as for distances in another class of graphs, derived from certain planar graphs by replacing their faces by cliques. However, we did not show that the class of spaces we studied was adequate to represent all planar tight spans, we did not provide algorithms for constructing the tight spans in the cases we studied, and we did not characterize the graphs with planar tight spans.

Our $O(n^2)$ time algorithm for two-dimensional $\ell_1$ and $\ell_\infty$ embedding improves a bound of $O(n^2\log^2 n)$ for the same problem by Edmonds~\cite{Edm-DCG-08}, which in turn improved previous bounds~\cite{MalMal-DCG-92}. Subsequently to our initial announcement of the research presented here, Catusse, Chepoi and Vax\`es~\cite{CatCheVax-TCS-11} modified Edmonds' algorithm to solve the Manhattan plane embedding problem directly in $O(n^2)$ time. However, their work does not solve the more general tight span construction problem that we do, nor does their algorithm have the incremental properties of ours.

The work presented here raises the question of whether there is a more general relation between dimension and complexity in tight spans: if the tight span consists entirely of low-dimensional features, does this imply an improved bound on the total complexity of its features? And if so, can we use this reduced complexity to construct tight spans more quickly? In recent work with Maarten L\"offler~\cite{EppLof-SoCG-11}, we were able to use the polyhedral approach to tight spans to show that the answer to both questions is yes. Whenever the dimension of the features in the tight span is bounded by a fixed constant, the tight span has complexity bounded by a fixed polynomial of its number of points and can be constructed in polynomial time. However, the exponent of this polynomial bound grows quadratically in the feature dimension. Finding tighter bounds on the complexity of the tight span as  a function of the dimension remains open.

\subsection{Organization}
The remainder of this paper is organized as follows.
In Section~\ref{sec:prelim} we provide more detailed definitions of hyperconvex metric spaces and tight spans, and we discuss results from our previous work~\cite{Epp-TA-09} on \emph{Manhattan orbifolds}, a class of two-dimensional hyperconvex metric spaces central to the work we present here. Section~\ref{sec:rep} describes \emph{rectangular complexes}, a data structure that is closely related to Manhattan orbifolds, but that provides a concrete representation of a two-dimensional hyperconvex metric space. In the same section we show how to augment these complexes with additional data that allows distance queries to be answered efficiently. In Section~\ref{sec:construct} we describe an incremental algorithm for building a rectangular complex that represents the tight span of a metric space; this algorithm succeeds whenever the given metric space has a planar tight span. In Section~\ref{sec:mpe} we characterize the rectangular complexes that may be embedded isometrically into the Manhattan plane, and we use this characterization together with our incremental tight span algorithm to provide an efficient algorithm for embedding a given distance matrix into the Manhattan plane, whenever an isometric embedding exists. We conclude with a discussion of our results in Section~\ref{sec:discuss}.

\section{Preliminaries}
\label{sec:prelim}
In order to provide a self-contained description of our work, we begin with a review of known material concerning hyperconvex metric spaces, tight spans, and Manhattan orbifolds.

\subsection{Hyperconvex metric spaces}

\begin{figure}[t]
\centering\includegraphics[width=5in]{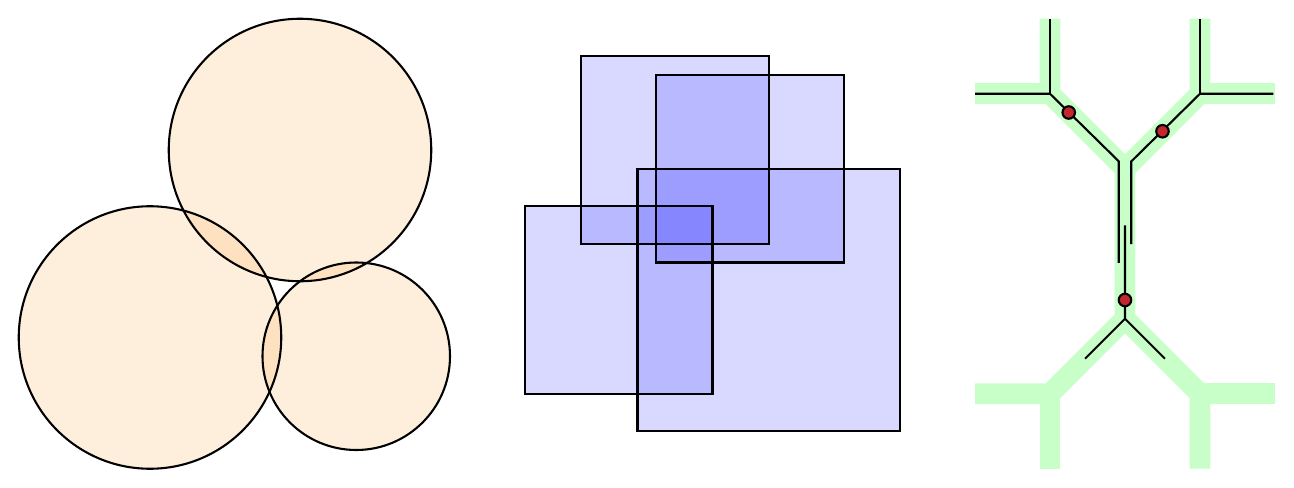}
\caption{In the Euclidean plane (left), balls may intersect pairwise without having a common intersection, but balls in the $\ell_\infty$ plane (center) and balls in real trees (right) have the Helly property.}
\label{fig:hyperconvex}
\end{figure}

A \emph{metric space} is a pair $(X,d)$, where $X$ is a set of \emph{points} and $d$ is a \emph{distance function}, a function from pairs of points to real numbers that is positive ($d(x,y)\ge 0$ with equality only for $x=y$), symmetric ($d(x,y)=d(y,x)$), and satisfies the triangle inequality ($d(x,y)+d(y,z)\ge d(x,z)$). Algorithmically, an arbitrary finite metric space with $n$ points may be represented by an $n\times n$ distance matrix or by shortest path lengths in a weighted undirected graph; in this paper, we generally prefer a distance matrix representation.

An \emph{isometry} from a metric space $(X,d)$ to a subset of another metric space $(X',d')$ is a function $f$ from $X$ to $X'$ such that, for all $x$ and $y$ in $X$, $d(x,y)=d'(f(x),f(y))$. A \emph{geodesic} in a metric space is the isometric image of a Euclidean line segment; a metric space is \emph{path-geodesic} if every two points are the endpoints of a geodesic.

A family $\mathcal{F}$ of sets has the \emph{Helly property} if every pairwise-intersecting subfamily has a common intersection. That is, if $\mathcal{S}\subset\mathcal{F}$ is a subfamily of $\mathcal{F}$, and every two sets in $\mathcal{S}$ have a nonempty intersection, it must be the case that there exists a point belonging to all sets in $\mathcal{S}$. A \emph{hyperconvex space} is a path-geodesic metric space in which the family of all closed balls $B_r(x)=\{y\mid d(x,y)\le r\}$ has the Helly property.
The Euclidean plane is not hyperconvex: three disks may intersect pairwise without having a common intersection (Figure~\ref{fig:hyperconvex}, left). However, $\ell_\infty$ spaces are hyperconvex (Figure~\ref{fig:hyperconvex}, center), as is the $\ell_1$ plane. Another familiar example of a hyperconvex space is a \emph{real tree}, a topological representation of a tree as a complex of points (vertices of the tree) and line segments (edges of the tree), with the property that any two points are connected by a unique geodesic, as depicted in Figure~\ref{fig:hyperconvex}, right).

\subsection{The tight span}
\label{sec:tight-span}

As Isbell~\cite{Isb-CMH-64} first showed,
any metric space $(X,d)$ (finite or infinite) may be embedded isometrically into a minimal hyperconvex metric space $T_X$ via a construction now known variously as the \emph{tight span}, \emph{injective envelope}, or \emph{hyperconvex hull}; although there are different ways of describing the tight span, they are equivalent in the sense that they give isometric metric spaces.

One way to construct the tight span is to let the set of points of $T_X$ be the set of all functions from $X$ to the set $\R$ of real numbers that obey the following two requirements:
\begin{itemize}
\item For each $x$ and $y$ in $X$, and for each $f\in T_X$, $f(x)+f(y)\ge d(x,y)$, and
\item For each $x$ in $X$, and for each $f\in T_X$, $\inf_{y\in X} f(x)+f(y)-d(x,y)=0$.
\end{itemize}
We may think of $f(x)$ as the distance from $f$ to $x$; with this interpretation, the first requirement is just the triangle inequality (the distance from $x$ to $f$ added to the distance from $f$ to $y$ is no shorter than the direct distance from $x$ to $y$); a special case of this first requirement, for $x=y$, forces all function values to be non-negative. The second requirement forces functions in $T_X$ to be minimal, in the sense that no function value $f(x)$ can be decreased without violating the triangle inequality.

As well as constructing the set of points in $T_X$, we need to define their distances. For the construction above, distance between points in $T_X$ is measured using the $\ell_\infty$ metric or sup-norm:
\[
d(f,g)=\sup_{x\in X} |f(x)-g(x)|.
\]
If $y$ is chosen arbitrarily from $X$, then it follows from the two requirements on functions in $T_X$ that, for any $x\in X$, $|f(x)-g(x)|\le f(y)+g(y)$.
Therefore, the supremum in this distance formula is bounded, and well-defined even when $X$ is infinite.

We may embed $(X,d)$ into $F$, by mapping any point $x$ to the function $f_x$ defined by the equation $f_x(y)=d(x,y)$. 
As shown by Isbell~\cite{Isb-CMH-64}, this embedding is isometric and $T_X$ is hyperconvex. More, if $X$ has an isometry into any hyperconvex metric space $(X',d')$, then that isometry can be extended to an isometry of $T_X$ into $(X',d')$, so in some sense $T_X$ is the smallest possible hyperconvex metric space into which $(X,d)$ can be isometrically embedded.

\begin{figure}[t]
\centering
\begin{minipage}[b]{0.3\linewidth}
\centering
\includegraphics[width=2in]{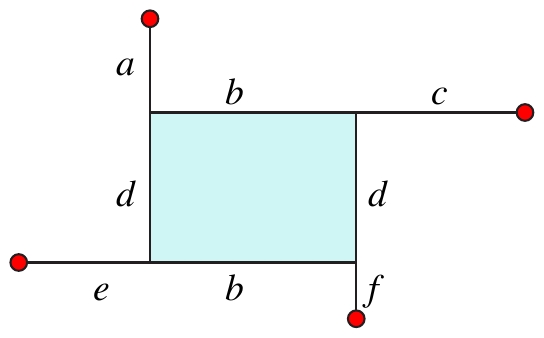}
\end{minipage}
\hspace{0.03\linewidth}
\begin{minipage}[b]{0.62\linewidth}
\bigskip{\small
\[
\left[\begin{array}{cccc}
0&a+b+c&a+b+d+f&a+d+e\\
a+b+c&0&c+d+f&b+c+d+e\\
a+b+d+f&c+d+f&0&b+e+f\\
a+d+e&b+c+d+e&b+e+f&0\\
\end{array}\right]
\]}
\vspace{0.25cm}
\end{minipage}
\caption{The tight span for metric spaces with four points can, in general, be represented as an axis-aligned rectangle in the $\ell_1$ plane, with line segments attached to its corners (left). On the right is the distance matrix for the four red points in this example.}
\label{fig:4-point-span}
\end{figure}

This construction of the tight span as a space of functions is general enough to apply to any metric space, but too general to be convenient algorithmically. For example, for metric spaces with exactly four points, the space of functions from $X$ to $\R$ is four-dimensional, and the tight span as constructed above is a subset of $\ell_\infty^4$; however, in this case a simpler two-dimensional representation of the tight span can always be found, in the form of an axis-aligned rectangle in the $\ell_1$ plane with line segments connecting its corners to the four points of the input metric space. This construction is depicted in Figure~\ref{fig:4-point-span}. (For four-point metric spaces that are not in general position, the rectangle or line segments may be degenerate, and this representation is not uniquely defined up to congruence.) Distances within this shape are inherited from the $\ell_1$ distance function in the plane containing the shape, and are isometric to the distances within the function-space definition of the tight span.
As in this example, in our algorithms we will  represent tight spans not as sets of functions, but rather as complexes of geometric objects (points, line segments, and rectangles) together with a definition of the distance between any two points in these complexes.

In order to show that our complexes correctly represent the tight span of the input metric space, we will show that they are hyperconvex, that they contain an isometric copy of the input space, and that, for each point $y$ in the complex, the function $f_y(x)=d(y,x)$ is minimal in the sense of the second requirement on the function spaces given above. If the space we construct is hyperconvex and contains the input space, then by the injectivity property of tight spans it must be a superset of the tight span. And if the space we construct has only minimal distance functions, it must be a subset of the tight span. Therefore, if we can show that the complexes we construct have all of these properties, they must be isometric to the tight span.

\subsection{Unique injectivity of one-point extensions of tight spans}

If $(X,d)$ is a metric space, and $Y\subset X$, then by the injective property of tight spans,
$T_Y$ can be embedded isometrically into $T_X$ in such a way that the embedding is the identity on the points of $Y$. However, this embedding may not be uniquely defined. For instance, if four points $abcd$ form a rectangle in the $\ell_1$ plane, their tight span is that rectangle, but the tight span of two opposite corners of the rectangle (a geodesic) may be embedded isometrically into the rectangle as any of infinitely many different monotone paths.

As we now show, when $X$ and $Y$ differ only in a single point, then this pathological behavior cannot occur: in this case $T_Y$ has a unique isometric embedding into $T_X$.

\begin{lemma}
\label{lem:unique-extension}
Let $(X,d)$ be any metric space, and $Y=X\setminus\{y\}$. Then there exists a unique isometric embedding from the tight span $T_Y$ of $Y$ to the tight span $T_X$ of $X$ that maps $Y$ to the image of $Y$ in $T_X$.
\end{lemma}

\begin{proof}
As in the previous section, represent each point of $T_Y$ as a function $f$ from $Y$ to $\R$ satisfying certain constraints, and represent each point of $T_X$ similarly. Then an embedding from $T_Y$ to $T_X$ may be specified by determining, for each function $f$ in $T_Y$, the value to assign to $f(y)$ to extend $f$ into a function in $F_X$. However, there is only one way to choose $f(y)$ in order to satisfy the constraints defining $T_X$: we must set $f(y)=\sup_{x\in X} d(x,y)-f(x)$.
\end{proof}

\subsection{Manhattan orbifolds}
\label{sec:mo}

The $\ell_\infty$ plane (or, isometrically, the $\ell_1$ plane) is a two-dimensional hyperconvex space, but not the only possible such space.
In previous work~\cite{Epp-TA-09} we defined a broader class of two-dimensional hyperconvex metric spaces,  \emph{Manhattan orbifolds}. In Manhattan orbifolds, every point except for a discrete set of \emph{singularities} looks locally like a point in a subset of the Manhattan plane.
For instance, the subset of the $\ell_\infty$ plane consisting of the points on or inside a Euclidean circle (Figure~\ref{fig:manhattan}, left) forms a Manhattan orbifold.

\begin{figure}[t]
\centering\includegraphics[width=5in]{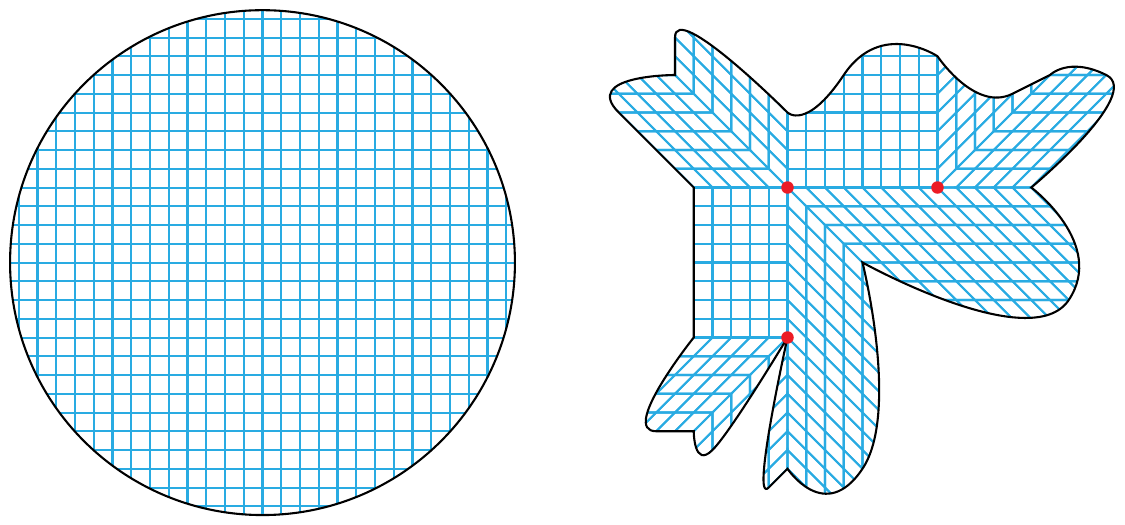}
\caption{Two examples of Manhattan orbifolds: the subset of the $\ell_\infty$ plane inside a circle (left), and a more complicated example (right) with two cone points of orders six and five, and one cone inflection point (the three red circles). The small blue squares and parallelograms mark regions that, within the intrinsic geometry of the Manhattan orbifold, are isometric to small squares in the $\ell_\infty$ plane.}
\label{fig:manhattan}
\end{figure}

The precise definition of a Manhattan orbifolds incorporates two different types of constraints, topological and metrical. Topologically, a Manhattan orbifold is required to be a 2-manifold with boundary, and every simple closed curve must form the boundary of a unique disk.
Metrically, a Manhattan orbifold must be Cauchy complete: if $p_0, p_1, \dots$ is a sequence of points such that, for every $\epsilon>0$, all but finitely many points are within distance $\epsilon$ of each other, then this sequence must have a limit point. Manhattan orbifolds must be path-geodesic: for every two points, the space contains a curve connecting them, with length equal to the distance between the points. And finally, every point in the space must have a neighborhood of one of five types, whose local geometry is modeled after the geometry of points in the $\ell_\infty$ plane:
\begin{itemize}
\item A point is an \emph{ordinary point} if it has a neighborhood that is isometric to a neighborhood of the origin in the $\ell_\infty$ plane. For instance, in the example of a circular subset of the $\ell_\infty$ plane, all of the points interior to the circle are ordinary points.
\item A point $p$ in a Manhattan orbifold is a \emph{boundary-geodesic point} if there is a region $R$ of the $\ell_\infty$ plane, bounded by a smooth curve, and a point $p'$ on the boundary of $R$, such that some neighborhood of $p$ in the orbifold is isometric to a neighborhood of $p'$ in $R$, and such that the slope of the boundary curve of $R$ near $p'$ has constant sign. For instance, in the example of a circular subset of the $\ell_\infty$ plane, all but four of the points on the circle (the four points where the circle has a horizontal or vertical tangent line) are boundary-geodesic.
\item A point $p$ is an \emph{Inflection point} if (like a boundary-geodesic point) it has a neighborhood that is isometric to the neighborhood of a point $p'$ on the boundary of a region $R$ in the $\ell_\infty$ plane bounded by a simple closed curve, but the boundary curve of $R$ near $p'$ is not smooth or has both points of both positive negative slope.  Within the given neighborhood, the other points are required to be either ordinary or boundary-geodesic. For instance, in the example of a circular subset of the $\ell_\infty$ plane, the four extreme points of the circle in each axis-parallel direction are inflection points.
\item \emph{Cone points} of order $k$ (for $k\ge 5$)  have a neighborhood isometric to the neighborhood of the origin in a metric space formed by gluing together $k$ quadrants of the Manhattan plane. In Figure~\ref{fig:manhattan}, right, the top left red circle marks a cone point of order six, and the top right red circle marks a cone point of order five.
\item \emph{Cone inflection points} have a neighborhood isometric to the neighborhood of the origin in a region of the order-$k$ rectilinear cone bounded by a curve. All other points in the neighborhood are required to be ordinary or boundary-geodesic. Figure~\ref{fig:manhattan} (right) has one cone inflection point, the lower red circle; the other points on the boundary of the figure are all either boundary-geodesic or inflection points.
\end{itemize}

A Manhattan orbifold, then, is defined to be a path-geodesic and Cauchy complete metric space, forming a 2-manifold with boundary in which every simple closed curve bounds a unique disk, such that every point in the space is one of these five types. The main result of our previous work~\cite{Epp-TA-09} is the following:

\begin{lemma}
\label{lem:manhattan-hyperconvex}
Every Manhattan orbifold is hyperconvex.
\end{lemma}

We will use this fact to prove the hyperconvexity of the rectangular complexes that we use to represent planar tight spans. The definition of rectangular complexes in Section~\ref{sec:rectplex} is modeled after that for Manhattan orbifolds, in that we require them to obey similar topological conditions and we impose similar constraints on the local geometry of the neighborhoods of points where rectangles are glued together. However, we will allow two-dimensional parts of the complex to be connected to each other by one-dimensional \emph{bridges} and zero-dimensional \emph{articulation points}, so that the whole complex will not necessarily form a 2-manifold with boundary. Additionally, although Manhattan orbifolds may have arbitrary piecewise-smooth boundary curves, we will require our rectangular complexes to have boundaries that are piecewise linear of restricted slopes ($\pm 1$ in the $\ell_\infty$ plane or axis-parallel in the $\ell_1$ plane).

We also need a technical lemma, adapted from a lemma in the same paper:

\begin{lemma}
\label{lem:four-corners}
Let $P$ be a polygonal simple closed curve in a Manhattan orbifold such that all sides of $P$ are parallel to a coordinate axis of the local $\ell_1$ geometry of the orbifold.
Then $P$ has at least four vertices with interior angle $\pi/2$.
\end{lemma}

\begin{proof}
Lemma~2 of~\cite{Epp-TA-09} is the same statement for polygons with sides parallel to the coordinate axes of the local $\ell_\infty$ geometry (equivalently, with slope $\pm 1$ in the local $\ell_1$ geometry). The same proof applies here: if $P$ is a simple closed curve as described in the lemma, we could glue two copies of the disk bounded by $P$ together along their boundary and replace the $\ell_1$ geometry of the copies by $\ell_2$ geometry, forming a surface that is topologically equivalent to a sphere, with a geometry that is locally Euclidean except at the points of singularity within $P$ and the vertices of $P$. At each point of singularity or vertex of $P$, we may define the \emph{angular defect} as the difference between $2\pi$ and the sum of the angles of the quadrants that meet at that point in terms of its local $\ell_2$ geometry. By Lemma~1 of~\cite{Epp-TA-09}, there are finitely many singularities within $P$. A form of the Gauss--Bonnet theorem states that the total angular defect of a two-dimensional manifold that (like the one here) is locally Euclidean except at a finite number of cone points is $2\pi$ times its Euler characteristic~\cite{Sch-MS-11}. Because the surface formed from two copies of $P$ is topologically a sphere, it has Euler characteristic~2 and total angular defect $4\pi$. A vertex of $P$ with interior angle $\pi/2$ leads to a point with only two quadrants in the glued surface (one for each copy of the disk), so it has a positive defect of $\pi$; all other vertices and singularities have zero or negative defect. Therefore, in order to achieve a total defect of $4\pi$, at least four vertices of $P$ must have interior angle $\pi/2$.
\end{proof}

\section{Representation by rectangular complexes}
\label{sec:rep}

In this section we define \emph{rectangular complexes} to be a certain kind of complex of points, line segments, and $\ell_1$-geometry rectangles, with constraints on how they are glued together that are analogous to the constraints of Manhattan orbifolds. We show that these complexes always define hyperconvex spaces, and we show how to augment them with additional information that allows distances within them to be calculated quickly.
In later sections we will show that any planar tight span of a finite point set can be represented by a complex of this type.

\subsection{Rectangular complexes}
\label{sec:rectplex}

We define a \emph{rectangular complex} to be a finite set of objects of three types: vertices, edges, and faces.
\begin{itemize}
\item
Each face of a rectangular complex represents a metric space with the geometry of an axis-aligned rectangle in the $\ell_1$ plane. Each of the four boundary sides of the rectangle either consists of a single edge, or is partitioned into
multiple edges.
\item
Each edge of a rectangular complex represents a metric space with the geometry of a line segment. An edge may lie on the boundaries of one or two faces; it is a \emph{bridge} if it does not lie along not the boundary of any face. Each edge has two endpoints, which are vertices of the complex. The total length of the set of edges along the side of any one of the faces of the complex must equal the length of that side.
\item
Each vertex of a rectangular complex represents a single point in a metric space. A vertex may be the endpoint of one or more edges, and it may lie on the boundaries of one or more faces.
\end{itemize}

\begin{figure}[t]
\centering\includegraphics[width=3.5in]{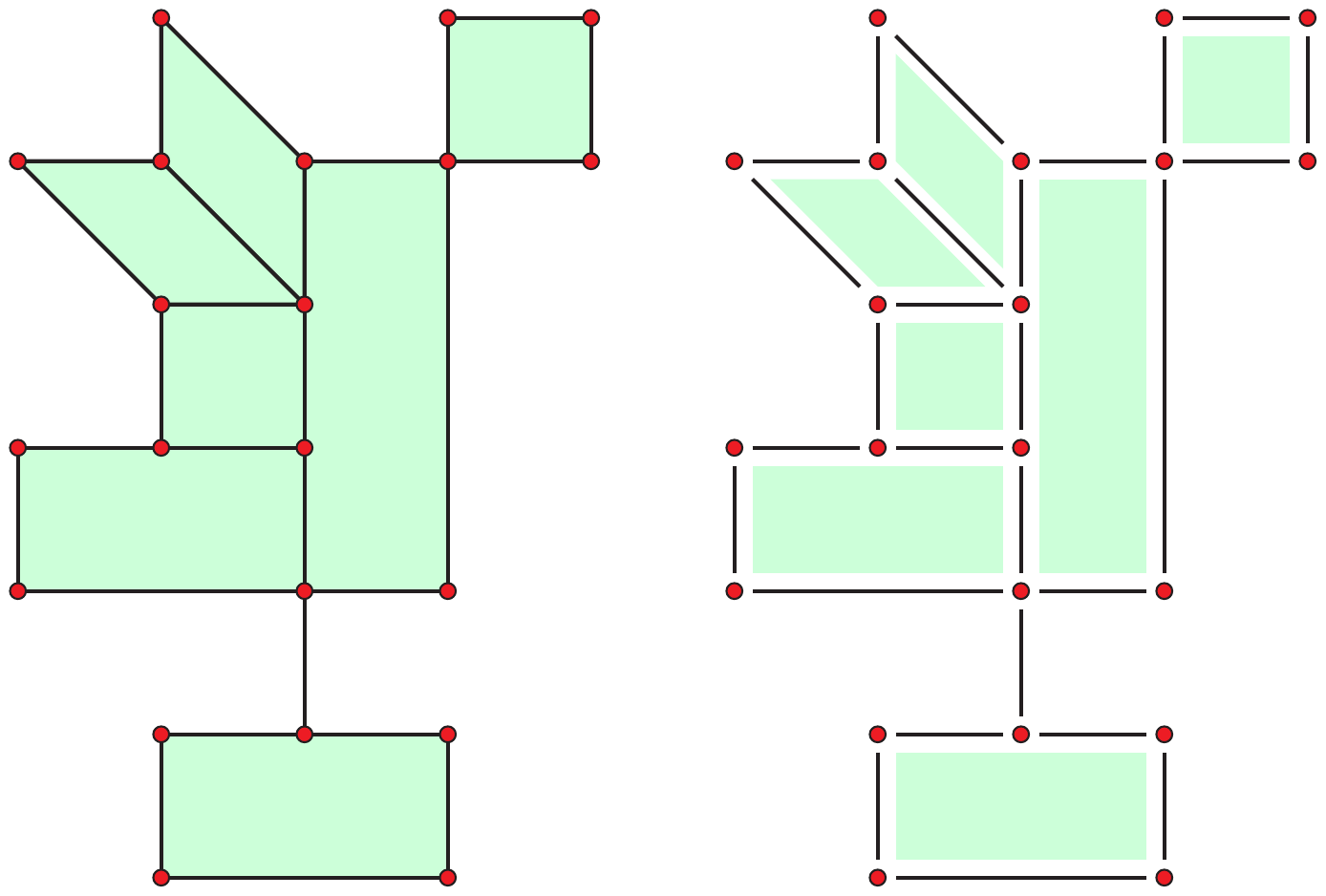}
\caption{An example of a rectangular complex, shown glued together (left) and in an exploded view with separate vertices, edges, and faces (right). This complex has 21 vertices, 27 edges, seven faces, three articulation points, and one bridge. Two of the faces are drawn as parallelograms, but should be interpreted as having rectangular geometry as part of the complex. There are two interior vertices; the upper one is a cone point of order~5, and the lower one is an ordinary point.}
\label{fig:rectangular}
\end{figure}

We say that an edge is incident to the faces for which it is part of the boundary, and to the vertices at its endpoints; a vertex is incident to a face if there is an edge incident to both the face and the vertex.
In order to ensure that a rectangular complex has the topology of a subset of the plane, and that its two-dimensional parts form Manhattan orbifolds, we require that at every vertex the incident faces have one of two possible structures:
\begin{itemize}
\item The incident edges and faces may form a single cycle, alternating between edges and rectangles, such that each edge in the cycle lies on the boundaries of the two rectangles that are adjacent to it in the cycle. In this case, topologically, the vertex has a neighborhood homeomorphic to an open disk in the plane. We define the \emph{angle} of an incident rectangle to be $\pi/2$ if the vertex is one of the four corners of the rectangle, and $\pi$ if the vertex instead lies on one of the rectangle's sides; we require that the total angle of the rectangles incident with the vertex be at least $2\pi$. In the terminology of Section~\ref{sec:mo}, the vertex must be either an ordinary point (with total angle $2\pi$) or a cone point (with total angle greater than $2\pi$).
\item Alternatively, it may be possible to group the incident edges and rectangles into one or more subsets, such that within each subset there is a sequence of alternating edges and rectangles that starts and ends on an edge and such that each edge in the sequence lies on the boundary of the adjacent rectangle or rectangles in the sequence. If there is exactly one of these subsets, and it includes at least one rectangle, then the vertex automatically forms a boundary-geodesic point, inflection point, or cone inflection point according to the local geometry of that subset. However, it is also possible for a vertex to be incident to a single bridge (in which case we call it a \emph{leaf}) or to be incident to more than one subsequence of alternating edges and rectangles (in which case we call it an \emph{articulation point}).
\end{itemize}

Additionally, as we did with Manhattan orbifolds, we require that any simple closed curve in the topological space obtained from the complex be the boundary of a unique disk, and that the space be connected. Figure~\ref{fig:rectangular} shows an example of a rectangular complex.

When we use rectangular complexes to represent tight spans, we will not require the points of the input metric space to be vertices of the complex. A point of the input metric space may be a vertex, but it may also be an interior point of an edge, or an interior point of a face.

\subsection{Computer representation of a rectangular complex}

In order to represent a rectangular complex within a computer algorithm, we assume an object-oriented representation in which each vertex, each edge, and each face of the complex is represented by an object of an appropriate type. Each of these objects has instance variables describing its connections to the other objects in the complex:
\begin{itemize}
\item
A face of the complex has the geometry of an $\ell_1$ rectangle. Associated with the object representing the face, we store the lengths of the sides of this rectangle, and, for each of the four rectangle sides, a partition of that side into a sequence of edges (represented as a sequence of pointers to edge objects).
\item
An edge of the complex has the geometry of a line segment.
Associated with the object representing the edge, we store its length, two pointers to the objects representing the vertices at its endpoints, and from zero to two pointers to the objects representing the faces that it forms part of the boundary of.
\item
A vertex of the complex has the geometry of a point.
Associated with the object representing the vertex, we store an unordered collection of the edges for which it is the endpoint.
\end{itemize}

In Section~\ref{sec:dist} we describe additional information that we associate with each of these types of object in order to compute distances quickly.

\subsection{Hyperconvexity of rectangular complexes}

We define the distance between two points in a rectangular complex to be the length of the shortest curve that connects them, where the length of a curve is defined to be the sum of its lengths within each of the faces and edges that it crosses, using the $\ell_1$ distance to measure lengths within each face. This distance function is well-defined by the assumption that rectangular complexes are connected. As we now show, this distance is hyperconvex.

To begin with, we classify pieces of the complex at a more coarse level of refinement than the level of faces and edges. We define a \emph{block} of a rectangular complex to be a maximal set of edges and faces that can be reached from each other through face-edge incidences, together with all of the vertices incident to these edges and faces. A vertex that is incident to two or more blocks is an \emph{articulation point}

\begin{lemma}
\label{lem:manhattan-block}
Each block must either be a bridge or a Manhattan orbifold.
\end{lemma}

\begin{proof}
If an edge does not lie on the boundary of any faces, it forms a bridge.
All remaining blocks include at least one face.

A block is path-geodesic by the definition of the distance function, and it is Cauchy-complete because it is a finite union of spaces that are themselves Cauchy-complete (the faces, edges, and vertices of the complex). The requirement that at most two rectangles share a boundary edge, and the restrictions on the way that rectangles and edges can be glued together at any vertex, force a block to be a 2-manifold with boundary, and the definition of a rectangle complex includes the constraint that any simple closed curve bounds a unique disk. The only remaining requirements for verifying that a block is a Manhattan orbifold are the ones stating that the points of the block have to have one of five constrained types of neighborhood. We can verify this with a small amount of case analysis:
\begin{itemize}
\item A point that is in the interior of a face of the complex is necessarily an ordinary point.
\item A point that is in the interior of an edge of the complex is either an ordinary point (if the edge lies on the boundary of two faces) or a boundary-geodesic point (if the edge lies on the boundary of a single face).
\item A vertex whose incident faces and edges form a cycle is either an ordinary point (if the total angle is $2\pi$) or a cone point (if the total angle is greater than $2\pi$.
\item At a vertex $v$ whose incident faces and edges do not form a cycle, exactly one of the paths of faces and edges incident to $v$ must belong to the block; for otherwise, a simple closed curve through two different parts of the block incident to $v$ could not bound a disk. If the total angle of these incident faces is $\pi$, the vertex is boundary-geodesic, and if the total angle is $\pi/2$ or $3\pi/2$, it is an inflection point. In the remaining cases, $v$~is a cone inflection point.
\end{itemize}
This case analysis completes the proof that a block that is not a bridge satisfies all of the requirements of a Manhattan orbifold.
\end{proof}

\begin{lemma}
\label{lem:block-tree}
The incidences between blocks and articulation points of the complex form a tree.
\end{lemma}

\begin{proof}
If it contained a cycle, there would exist a simple closed curve through that cycle, which could not be the boundary of a disk in the complex.
\end{proof}

\begin{lemma}
\label{lem:1-sum}
Let $A$ and $B$ be two hyperconvex spaces, and let $A\cap B$ be a single point. Then $A\cup B$ (with the distance function defined by arc length in $A$ added to arc length in $B$) is hyperconvex.
\end{lemma}

\begin{proof}
Let $S$ be a set of metric balls that intersect pairwise. If each ball in $S$ contains the point $A\cap B$, they have a common intersection. Otherwise, at least one of the balls must be entirely within one of the two sets; without loss of generality it is entirely within $A$. Then the restriction of all of the balls to $A$ forms a set of metric balls within $A$, for each ball centered at a point in $A$ remains a ball in $A$ and each ball centered at a point in $B$ restricts to a ball in $A$ centered at $A\cap B$. These balls still must intersect pairwise, for the balls that have points in $B$ all intersect at $A\cap B$ and the remaining pairwise intersections are not affected by the restriction to $A$. Therefore, by hyperconvexity of $A$, $S$ has a common intersection. We have shown that any set of balls in $A\cup B$ with pairwise intersections has a common intersection, so by definition $A\cup B$ is hyperconvex.
\end{proof}

\begin{theorem}
\label{thm:rect-complex-is-hyperconvex}
Any rectangular complex defines a hyperconvex metric space.
\end{theorem}

\begin{proof}
By Lemmas \ref{lem:manhattan-hyperconvex}~and~\ref{lem:manhattan-block}, every block is hyperconvex. By Lemma~\ref{lem:block-tree}, the whole complex may be obtained by gluing together blocks at single points (the articulation points of the complex). But connecting two hyperconvex spaces together by identifying a single point of one with a single point of another preserves hyperconvexity (Lemma~\ref{lem:1-sum}), so it follows that every rectangular complex represents a hyperconvex metric space.
\end{proof}

\subsection{Distances in rectangular complexes}
\label{sec:dist}
\def\site{\mathop{\mathrm{site}}}

The previous section described rectangular complexes as data structures for representing abstract metric spaces, but we are interested in them as a way to represent more specifically the tight spans of finite metric spaces. In this section we augment the rectangular complex to represent its relationship with the given finite metric space, and we describe how to compute distances in this augmented space. For clarity, we adopt terminology used in computational geometry in the context of Voronoi diagrams: a \emph{site} is one of the points of the finite metric space $(X,d)$ that we are attempting to find the tight span of (or, by extension, its image in the tight span) while a \emph{point} may refer to any point of the rectangular complex under discussion.

\begin{figure}[t]
\centering\includegraphics[width=\textwidth]{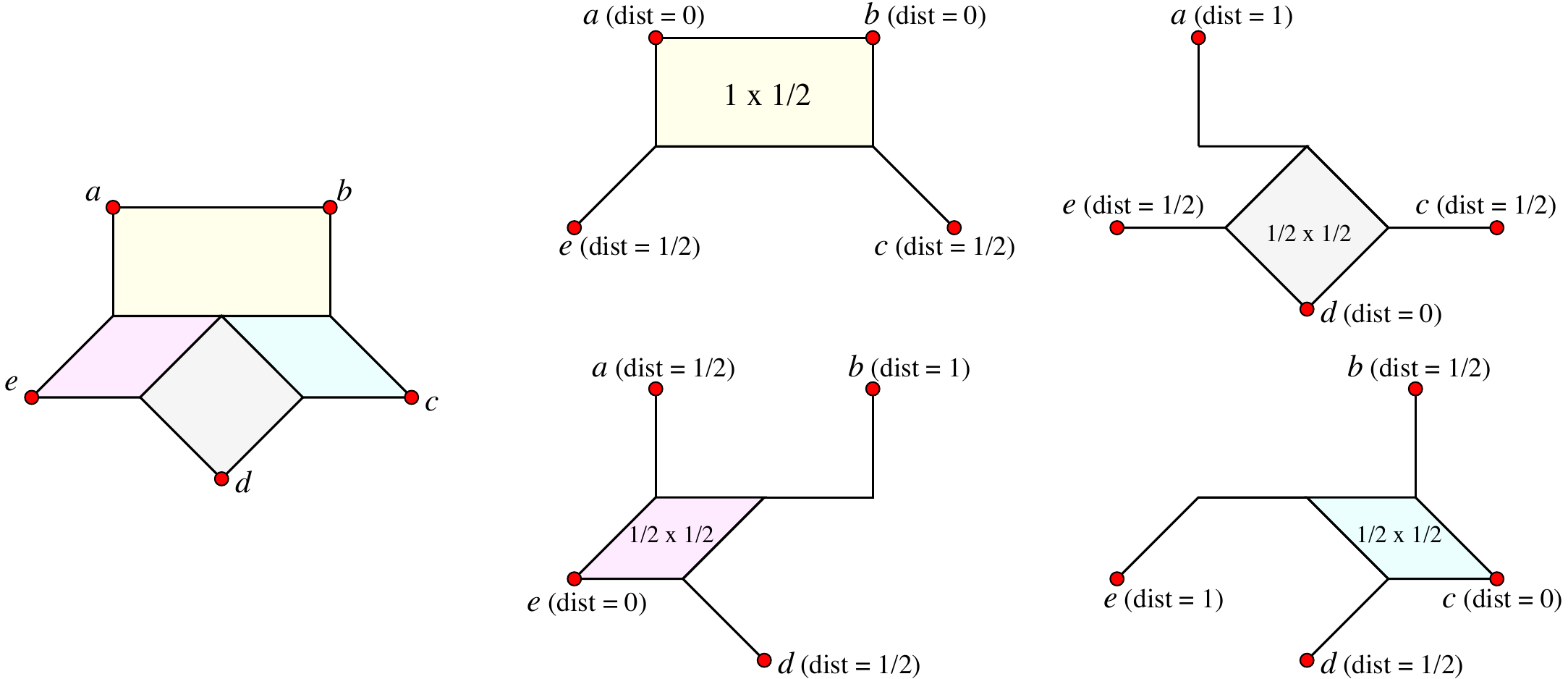}
\caption{One possible representation of the tight span of a 5-cycle $abcde$ (with distance 1 between adjacent sites in the cycle and distance 2 between nonadjacent sites) as a rectangular complex (the leftmost part of the figure), with an exploded view on the right showing the information stored in each face of the complex: its dimensions as an $\ell_1$ rectangle, and the distances to four representative sites, one for each corner of the rectangle. Because of the order-5 cone point in the center of the complex, this example forms a Manhattan orbifold but does not embed into the Manhattan plane.}
\label{fig:c5-exploded}
\end{figure}

Suppose that the rectangular complex $R$ represents the tight span of a finite set $(X,d)$ of sites. We may represent the isometry from $(X,d)$ to $R$ explicitly, by specifying for each site in $X$ which feature of $R$ it maps to and (if that feature is an edge or face) what its position within that feature is. If a site belongs to multiple features, we choose a single one of them arbitrarily.
We also store, as part of our data structure, two \emph{representative sites} from $X$ for each bridge of the rectangular complex and four representative sites from $X$ for each face of the rectangular complex, with the property that each bridge or rectangular face is a subset of the tight span of its representative sites. We associate these representative sites with the endpoints or corners of the bridge or rectangle: for each point $v_i$ that is an endpoint of a bridge or the corner of a rectangle, we store a representative site $\site(v_i)$ with the property that, for every point $r$ in the bridge or rectangle, $d(r,\site(v_i))=d(r,v_i)+d(v_i,\site(v_i))$. We also store, for each endpoint of a bridge or each corner of a rectangle, the distance $d(v_i,\site(v_i))$. If a single point of the rectangular complex is the corner or endpoint of more than one bridge or rectangle, it may have different representative sites and distances stored with it for each of the bridges and rectangles it belongs to.  We will show later that, as our tight span construction algorithm builds a rectangular complex representing the tight span of a set of sites, it is always able to find representative sites with the properties specified above.

From this information, we may calculate the distance between any point of the rectangular complex and any new site, using only information that can be looked up in constant time, via a formula resembling a landmark-based technique of Goldberg and Harrelson for bounding distances in arbitrary metric spaces~\cite{GolHar-SODA-05}:

\begin{lemma}
\label{lem:rect-distance}
Let $C$ be a rectangular complex, augmented as above, that represents the tight span of a set $S$ of $n-1$ sites, let $R$ be a rectangular face of $C$, let $r$ be a point in $R$, and let $s$ be an additional site. Let the four corners of $R$ be the points $v_0$, $v_1$, $v_2$, and $v_3$.
Then, in the tight span of $S\cup\{s\}$, 
\[
d(r,s)=\max \left\{ d(s,\site(v_i)) - d(v_i,\site(v_i)) - d(r,v_i)\mid 0\le i\le 3\right\}.
\]
\end{lemma}

\begin{proof}
By Lemma~\ref{lem:unique-extension}, $C$ has a unique isometric embedding into the tight span of $S\cup\{s\}$, and the distance $d(r,s)$ is well-defined as the distance between the images of $r$ and $s$ in this embedding. Let $T$ be the tight span of $\{s,v_0,v_1,v_2,v_3\}$.  Then (by injectivity) $T$ may be embedded isometrically into the tight span of $S\cup\{s\}$, from which it follows that the distance between $r$ and $s$ as measured within $T$ is the same as the distance between $r$ and $s$ as measured in the tight span of $S\cup\{s\}$.

Recall that (as with any tight span) $T$ may be defined as a space of functions from its set of defining sites $\{s,v_0,v_1,v_2,v_3\}$ to the real numbers and that, if the point $r$ in $T$ is represented in this way as a function $f_r$, then the distance $d(r,s)$ that we wish to calculate is just $f_r(s)$.
Recall also that one of the two defining properties of the tight span is that, for any function $f$ in this space of functions, and any two sites $x$ and $y$, it must be the case that $f(x)+f(y)\ge d(x,y)$. Applying this property with $f=f_r$, $x=s$ and $y=v_i$ gives the inequality $d(r,s)\ge d(s,v_i)-d(r,v_i)$. The triangle inequality $d(s,\site(v_i))\le d(s,v_i)+d(v_i,\site(v_i))$ allows us to replace the term $d(s,v_i)$ in this inequality, giving us the desired lower bound $d(r,s)\ge d(s,\site(v_i))-d(v_i,\site(v_i))-d(r,v_i).$

Next, let $U$ be the tight span of $\{s,\site(v_0),\site(v_1),\site(v_2),\site(v_3)\}$. Since $R\subset U$, $r$ belongs to $U$. The second of the two defining properties of the tight span of a finite set of sites is that, for any function $f$ in the space of functions defining $U$, and any site $x$, there exists a second site $y$ such that $f(x)+f(y)=d(x,y)$. Applying this property with $f=f_r$ and $x=s$, we obtain that there exists an index $i$ for which $d(r,s)=d(s,\site(v_i))-d(r,\site(v_i))$. However, $d(r,\site(v_i))= d(v_i,site(v_i))+d(r,v_i)$, and substituting this into the previous equality gives us $d(r,s)= d(s,\site(v_i))-d(v_i,\site(v_i))-d(r,v_i)$ as desired.
\end{proof}

The same proof shows an analogous result for the bridges of a rectangular complex.

\begin{lemma}
\label{lem:edge-distance}
Let $C$ be a rectangular complex, augmented as above, that represents the tight span of a set $S$ of $n-1$ sites, let $B$ be a bridge of $C$, let $r$ be a point in $B$, and let $s$ be an additional site. Let the two endpoints of $B$ be $v_0$ and $v_1$.
Then, in the tight span of $S\cup\{s\}$, $d(r,s)=\max_i d(r,\site(v_i)) - d(v_i,\site(v_i)) - d(r,v_i).$
\end{lemma}

\begin{corollary}
Suppose that $(X,d)$ is a finite metric space represented by a distance matrix, that $Y\subset X$, that $C$ is a rectangular complex that represents the tight span of $Y$, and that $C$ is augmented as above. Then, for any point $r$ in $C$ and any point $s$ in $X$ we can compute the distance $d(r,s)$, as measured in the tight span of $Y\cup\{s\}$, in time $O(1)$.
\end{corollary}

\section{Incremental construction algorithm}
\label{sec:construct}
In this section we will describe our algorithm for constructing planar tight spans, but first we need some more technical lemmas.

\subsection{Tight spans in the Manhattan plane}

\begin{figure}[t]
\centering\includegraphics[height=1.75in]{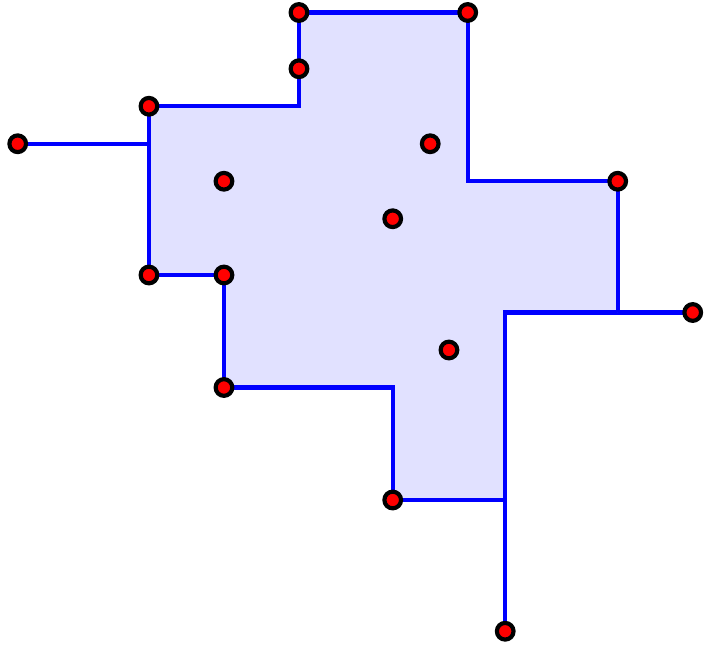}
\caption{The orthogonal hull of a set of point sites. When the hull is connected, it is isometric  to the tight span of the $\ell_1$ distances among its sites. Public-domain image drawn by the author for Wikipedia.}
\label{fig:orthogonal-hull}
\end{figure}

In order to prove the extension property in the next section, we need to study tight spans in the Manhattan plane. We say that a point $h$ in the $\ell_1$ plane is \emph{surrounded} by $X$ if each of the four closed axis-aligned quadrants with $h$ as their apex contains at least one point of $X$, and we define the orthogonal convex hull $H$ of $X$~\cite{KarOve-BIT-88,MonFou-TR-82,NicLeeLia-BIT-83,OttSoiWoo-IS-84} to consist of all points surrounded by $X$ (Figure~\ref{fig:orthogonal-hull}).

\begin{lemma}
\label{lem:L1-span}
Let $X$ be a nonempty subset of the $\ell_1$ plane (not necessarily finite).  If the orthogonal convex hull $H$ of $X$ is connected, then $H$ is isometric to the tight span of $X$.
\end{lemma}

\begin{proof}
Since $X$ is an isometric subset of the $\ell_1$ plane, which is a hyperconvex space, and since (by injectivity of tight spans) any isometric embedding of $X$ into a hyperconvex space can be extended to an isometric embedding of the tight span $T_X$ into the same hyperconvex space, it follows that the tight span of $X$ is isometric to some subset of the $\ell_1$ plane. That is, if we identify the tight span with the set of vectors describing distances from points in the span to sites in $X$, we need only consider distance vectors determined in this way from points in the $\ell_1$ plane. Also observe that if $q$ belongs to the orthogonal convex hull, then no other point $r$ of the $\ell_1$ plane can determine the same distance vector: for, if $p$ is a point of $X$ in the quadrant that has $q$ as apex and that is opposite from the quadrant containing $r$, then $d(p,q)\ne d(p,r)$.

We consider cases for a point $p$ of the $\ell_1$ plane, according to whether it is surrounded and which of the quadrants around it is nonempty:
\begin{itemize}
\item
If a point $p$ in $H$ is surrounded by $X$, its vector of distances to $X$ satisfies the requirements of the tight span that, for each site $q$ in $X$ there exists an $r$ such that $d(p,q)+d(p,r)=d(q,r)$: namely, take $r$ to be a point in the opposite quadrant from $q$. Therefore, its distance vector belongs to the tight span. Since that distance vector cannot be generated by any other point of the $\ell_1$ plane, it follows that $p$ must belong to the tight span of $X$.
\item
Suppose that a point $p$ is not surrounded, and that one of the empty quadrants for $p$ has an opposite non-empty quadrant. In this case, let $q$ be a point of $X$ in the nonempty quadrant. Then there can be no $r$ in $X$ for which $d(p,q)+d(p,r)=d(q,r)$; for every $r$ with this property is in the empty quadrant, which by assumption contains no points of $X$. Therefore, in this case, $p$ cannot belong to the tight span of $X$.
\item
Finally, if a point $p$ is not surrounded, and each of its empty quadrants has an opposite empty quadrant, then (since $X$ is nonempty) it must be the case that $p$ has two diagonally opposite empty quadrants and two diagonally opposite nonempty quadrants. But then the orthogonal convex hull is a subset of the interiors of the two nonempty quadrants and is not connected.
\end{itemize}
Thus, in each case it can be shown that either the convex hull coincides with the tight span or the convex hull is not connected.
\end{proof}

\subsection{One-point extension of geodesic paths}

Our algorithm for constructing planar tight spans is \emph{incremental}: it proceeds by adding one point to the input at a time, maintaining as it does the tight span of the points added so far. Therefore, we need to describe the ways in which the tight span may change when a single point is added. We have already seen one result in this direction, Lemma~\ref{lem:unique-extension}, stating the fact that for any metric space $X$ and point $s$ in $X$  (regardless of planarity) there is a unique embedding of the tight span of $X\setminus\{s\}$ into the tight span of $X$. Our next extension property concerns the case when the space $X\setminus\{s\}$ to which we are adding the point $s$ is a geodesic path; in this case, we show that the tight span of $X$ itself is necessarily planar.

Thus, let $P$ be a metric space that is isometric to a line segment. Then $P$ is hyperconvex and therefore it is its own tight span. If $s$ is an additional point, not on $P$, then there are many possible metric spaces on the set $P\cup \{s\}$; a particular choice in this set of possible metric space is determined (up to isometry) by the function that maps points of $P$ to their distance from $s$, and therefore the shape of the tight span of $P\cup\{s\}$ is determined by this function. In our application of these concepts, a particularly important choice of metric space on $P\cup\{s\}$ will be the spaces determined by the \emph{sawtooth functions}, which we define to be the functions $f$ from $P$ to the positive real numbers with the following properties:
\begin{itemize}
\item If we identify $P$ isometrically with a subset of the real line, then $f$ is piecewise linear,
\item the number of breakpoints of piecewise linearity is finite, and
\item for any subset  of the real line within which $f$ is linear, the derivative of $f$ is $\pm 1$.
\end{itemize}

\begin{lemma}
\label{lem:sawtooth}
Let $(X,d)$ be a metric space, $s\in X$, and $Y=X\setminus\{s\}$. Suppose that the tight span of $Y$ is represented by a rectangular complex $C$. Then along any path formed by the edges of $C$, the distance from $s$ forms a sawtooth function.
\end{lemma}

\begin{proof}
This follows immediately from the distance formula in Lemma~\ref{lem:edge-distance}.
\end{proof}

\begin{figure}[t]
\centering\includegraphics[height=2.5in]{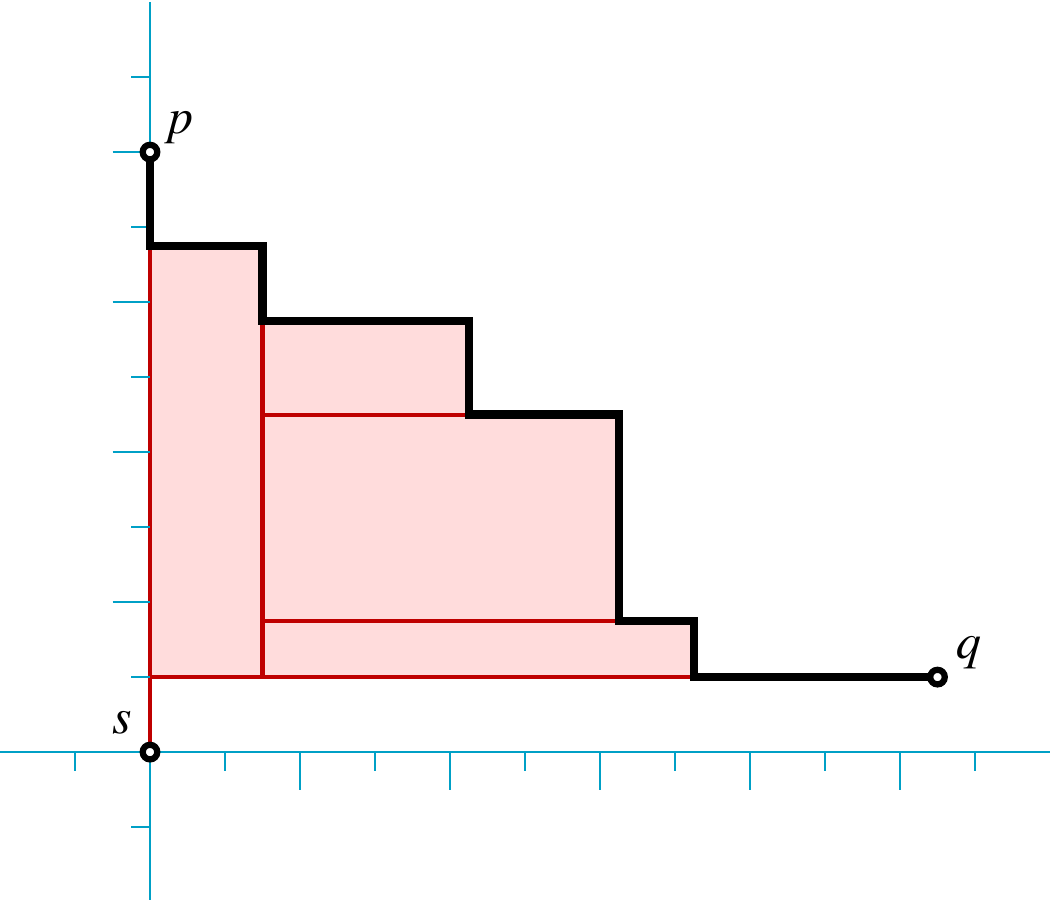}
\caption{Illustration for Lemma~\ref{lem:path-extension}, showing the embedding of path $pq$ in the upper right quadrant of the $\ell_1$ plane (heavy black edges), the additional point $s$ at the origin, and the tight span formed as the convex hull of the path and $s$ (medium red). In this example, the distances from the path to $s$ form a sawtooth function and the tight span has accordingly been partitioned into rectangles, such that only two rectangles lie on the boundary of the span.}
\label{fig:sawtooth}
\end{figure}

\begin{lemma}
\label{lem:path-extension}
Let $(X,d)$ be a metric space, $P$ a geodesic path in $X$, and $s$ a point in $X$ that is not on $P$.
Then the tight span of $P\cup\{s\}$ can be isometrically embedded into an orthogonally convex subset of the $\ell_1$ plane bounded between three geodesic paths, one of which is an isometric image of $P$ and the other two of which connect $s$ to the endpoints of the image of $P$.

If, in addition, the function that maps points of $P$ to their distance from $s$ is a sawtooth function, then the tight span of $P\cup\{s\}$ can be partitioned into a set of $\ell_1$ rectangles, together with $P$ and possibly a single line segment, such that  number of rectangles is equal to one less than the number of local minima of the sawtooth function, and such that the boundary of the tight span is covered by $P$,  by two of the rectangles, and by the line segment.
\end{lemma}

\begin{proof}
Let the endpoints of $P$ be the points $p$ and $q$, and map each point $r$ of $P$ to the point
\[
(x_r,y_r) = \left(\frac{d(p,r)+d(r,s)-d(p,s)}{2},\frac{d(p,s)-d(p,r)+d(r,s)}{2}\right)
\]
of the $\ell_1$ plane, as depicted in Figure~\ref{fig:sawtooth}. It follows from the triangle inequality that this point lies within the positive quadrant of the plane, and that the distance between the images of any two points in $P$ is equal to their distance in $P$, so this is an isometric embedding. The isometry may be extended to $P\cup\{s\}$ by mapping $s$ to the origin. In the case that the function from $P$ to its distances from $s$ is a sawtooth function, the image of $P$ is a polygonal chain in which each edge is axis-parallel.
Now let $T$ be the orthogonal convex hull of the image of $P$ and of the origin. $T$ is connected: it has the form of an orthogonal polygon bounded to the left by the vertical line segment $ps$ along the $y$ axis, from above and to the right by $P$ itself, and from below and to the right by a path from $s$ to $q$ that follows the $y$ axis and the horizontal line through $q$. Therefore, by Lemma~\ref{lem:L1-span}, it is the tight span.
In the case of a sawtooth function, we may cover $T$ with rectangles as shown in the figure.
\end{proof}

\subsection{Attachment points}
Let $(X,d)$ be a metric space, $s\in X$, and $Y=X\setminus\{s\}$. Let $T_X$ and $T_Y$ be the tight spans of $X$ and $Y$ respectively. We define a point $t\in T_Y$ to be a \emph{weak attachment point} for $s$ if every neighborhood of $t$ in $T_X$ contains a point in $T_X\setminus T_Y$, and a \emph{strong attachment point} if every geodesic from $t$ to $s$ in $T_X$ lies entirely within $T_X\setminus T_Y$. Observe that a strong attachment point is automatically a weak attachment point.

In a rectangular complex, we can characterize the attachment points much more specifically.

\begin{lemma}
\label{lem:rect-strong-attachments}
Let $X$ be a metric space and $s$ be a point in $X$.
Let $C$ be a rectangular complex that represents the tight span of $X\setminus\{s\}$.
Then for each point $p$ in $C$, either $p$ is a local minimum of distance from $s$ or there exists a geodesic $P$ in $C$ such that $p$ is an endpoint of $P$ and such that, for any two points $t$ and $t'$ of $P$,
\[
d(p,t)+d(t,s)=d(p,t')+d(t',s).
\]
A point of $C$ is a strong attachment point for $s$ in $C$ if and only if it is a local minimum of distance from $s$.
\end{lemma}

\begin{proof}
If $p$ is not a local minimum of distance, let $c$ be the edge or rectangle of $C$ that contains $p$ and in which $p$ is not a local minimum, let $q$ be the site in $X\setminus\{s\}$ that (according to the distance formulae in Lemmas~\ref{lem:rect-distance} and~\ref{lem:edge-distance}) determines the distance from $p$ to $s$, let $P'$ be a geodesic from $p$ to the vertex of $c$ farthest from $q$, and let $P$ be the portion of $P'$ within which $t$ determines the value of the distance formula.
Then for any point $t$ within $q$, $d(p,t)+d(t,s)=d(q,t)-d(q,p)+d(t,s)=d(q,s)-d(q,p)$. Since the right hand side of this formula does not depend on $t$, it is equal for any two different points $t$ and $t'$ of $P$.

If such a geodesic $P$ exists, then it can be extended into a geodesic in $T_X$ from $p$ to $s$, showing that $p$ is not a strong attachment point. However, if $p$ is a local minimum of distance, there can exist no such geodesic within $C$, so every geodesic from $p$ to $s$ in $T_X$ must avoid $C$, and $p$ is a strong attachment point.
\end{proof}

We now use Lemma~\ref{lem:path-extension} to show a weakened form of path-connectivity for the attachment points in any metric space.

\begin{lemma}
\label{lem:attachment-geodesic}
Let $X$ be a metric space, let $s$ be a point of $X$, and let $Y=X\setminus\{s\}$.
Let $p$ and $q$ be strong attachment points for $s$ in the tight span $T_Y$ of~$Y$. Then there exists a geodesic from $p$ to $q$ in $T_Y$ in which all points are weak attachment points.
\end{lemma}

\begin{proof}
Let $R$ be a maximal set of strong attachment points that all lie on a single geodesic $P$ from $p$ to $q$ in $T_Y$; then $P$ may be partitioned into intervals within which the strong attachment points are dense and intervals within which there are no strong attachment points. The set of weak attachment points is closed; therefore, if $I$ is an interval of $P$ within which the strong attachment points are dense, every point of $I$ must be at least a weak attachment point. Thus, we need only worry about the intervals within which there are no strong attachment points. Assume without loss of generality that $p$ and $q$ bound a single such interval; that is, $p$ and $q$ are strong attachment points such that there are no strong attachment points on any geodesic from $p$ to $q$.

Now let $T_{Ps}$ be the tight span of $P\cup\{s\}$, embedded into the $\ell_1$ plane with $s$ at the origin as described by Lemma~\ref{lem:path-extension}.
Some of the points of $T_{Ps}$ (including all points in $P$) belong to $T_Y$, and some other points (including $s$) do not belong to $T_Y$. Let $R$ be the axis-aligned bounding rectangle of $T_{Ps}$ in the $\ell_1$ plane. If the interior of $R$ contains any points of $T_Y$, let $r$ be such a point that is as close as possible to $s$; then $r$ must be a strong attachment point, contradicting the assumption that there are none on any geodesic from $p$ to $q$. If, on the other hand, all points in $T_{Ps}\cap T_Y$ lie on the boundary of $R$, then in particular the path $P$ follows this boundary and every point in $P$ is a weak attachment point.
\end{proof}

\subsection{Good paths}

We say that a point of a rectangular complex belongs to the \emph{boundary} of the complex if it belongs to a bridge or to an edge that is adjacent to only one rectangle, or if it is the endpoint of an edge of either of these two types.

Let $C$ be a rectangular complex, and let $S$ be a subset of its boundary. Then we define a \emph{good path} for $S$ to be a geodesic $P$ in $C$ with the following properties:
\begin{itemize}
\item $P$ starts and ends at a point of $S$, and contains all of the points of $S$.
\item Every point of $P$ belongs to the boundary of $C$.
\item If $P$ contains an articulation point $p$ of $C$, and $p$ is not an endpoint of $P$, then the points of $P$ on the two sides of $p$ belong to different blocks of $C$.
\item If $R$ is a rectangle of $C$, and $P$ contains points in the relative interiors of two adjacent sides of $R$, then the corner where these two sides meet belongs to $S$.
\end{itemize}
In our incremental algorithm, $S$ will be a set of strong attachment points and the good path will form a set of weak attachment points at which new parts of the tight span will be connected to the existing complex.

Let $X$ be a metric space and $s$ be a point in $X$.
Let $C$ be a rectangular complex that represents the tight span of $X\setminus\{s\}$, let $S$ be the set of local minima in $C$ for distance from $s$, and suppose that there exists a good path $P$ for $S$ in $C$. We say that $P$ is a \emph{usable path} if it satisfies the following additional constraint: for every point $p$ that is a local maximum $p$ of the sawtooth function on $P$ that is determined by the distance from $s$ (according to Lemma~\ref{lem:sawtooth}), either $p$ is an articulation point of $C$, or the rectangles incident to $p$ have a total angle of at least $3\pi/2$.

\begin{lemma}
\label{lem:attachment-is-good}
Let $X$ be a metric space and $s$ be a point in $X$.
Let $C$ be a rectangular complex that represents the tight span of $X\setminus\{s\}$, and suppose that the tight span $T_X$ of $X$ itself is homeomorphic to a subset of the plane. Let $S$ be the set of strong attachment points for $s$ in $C$, and let $W$ be the set of weak attachment points. Then $W$ contains a usable path for $S$ in $C$.
\end{lemma}

\begin{proof}
By Lemma~\ref{lem:attachment-geodesic}, $W$ contains a geodesic between any two points of $S$. We observe that, for any three points of $S$, the three geodesics connecting them must have a single geodesic between two of them as their union, with the third point interior to this geodesic: for, otherwise, some two of the geodesics would have a union that is not a path, and the two-dimensional subsets of $T_X$ that (by Lemma~\ref{lem:path-extension}) are attached to these geodesics would form a nonplanar set. By the same argument, if $P$ is a geodesic subset of $W$ that covers two or more points in $S$ and is as long as possible, it must cover all of $S$, for otherwise the union of $P$ with a geodesic in $W$ from an uncovered point to an uncovered point at or near one of the ends of $P$ would violate maximality or planarity.

Then $P$ must contain only boundary points of $C$, for if it contained an interior point $r$ then $C$ together with a path in $T_X\setminus C$ from $r$ to $s$ would form a nonplanar set. Similarly, it is not possible for $P$ to pass through an articulation point but for both sides of the path near the point to remain within the same block, for this would cause a nonplanarity.

Finally, suppose that $P$ passes through two sides of a rectangle of $C$. Then, since $P$ is a geodesic and follows the boundary of $C$, it must pass through the corner point $c$ where the two sides meet.  And, if $c$ were not a strong attachment point, then by the distance formula in Lemma~\ref{lem:edge-distance} the distance from $s$ would be, near $c$, a sawtooth function for which it is not a local minimum. By Lemma~\ref{lem:path-extension}) within a neighborhood of $C$, the tight span of $P$ and $s$ would form either a single quadrant of the $\ell_1$ plane (if the distance from $s$ is a local maximum) or a pair of quadrants (if it is neither a local minimum nor a local maximum). But together with the rectangle itself this would form a cone point of order two or three, which is not hyperconvex, so additional points would have to be added to the tight span near $C$ to make a hyperconvex set, again violating planarity.

Thus, $P$ is a geodesic that contains all of the strong attachment points and satisfies all the other conditions of a good path, as required. It must be a usable path as well, because any violation of usability would lead to a cone point of order three and a violation of planarity as before.
\end{proof}

\begin{lemma}
\label{lem:good-is-tight}
Let $X$ be a metric space and $s$ be a point in $X$.
Let $C$ be a rectangular complex that represents the tight span of $X\setminus\{s\}$, let $S$ be the set of local minima in $C$ for distance from $s$, and suppose that there exists a usable path $P$ for $S$ in $C$. Let $D$ be the rectangular decomposition of the tight span of $P\cup\{s\}$ given by Lemma~\ref{lem:path-extension}. Then the union $C\cup D$ (with the sides of rectangles along $P$ subdivided into edges as appropriate) is a rectangular complex that represents the tight span of $X$.
\end{lemma}

\begin{proof}
By the definition of a good path, $C\cup D$ attaches the rectangles in $D$ to the boundaries of rectangles in $C$, so it has at most two rectangles per edge as is required in a rectangular complex. The conditions defining a good path ensure that each vertex continues to have the correct local neighborhood structure: the requirement that it pass to a new block of $C$ whenever it passes through an articulation point ensures that, at each such point, the neighborhood of the point continues to form a set of rectangle-edge paths. And the requirement that the path must form a strong attachment point whenever it passes through two adjacent sides of a rectangle ensures that, at each such point, the neighborhood has total angle $2\pi$. At any local maximum $p$ of distance from $s$, the requirement that $p$ be an articulation point or a point of total angle at least $3\pi/2$ means that, in $C\cup D$, either $p$ belongs to an alternating sequence of rectangles and edges (in which case we need not worry about their total angle) or it belongs to an alternating cycle with total angle at least $2\pi$. And at any other point of the path, the neighborhood has total angle equal to its total angle in $C$ (at least $\pi$ since it is not the corner of a rectangle) plus its total angle in $D$ ($\pi$ for a weak attachment point or $3\pi/2$ for a strong one), satisfying the conditions of a rectangular complex.

Thus, $C\cup D$ is a rectangular complex, and therefore a hyperconvex space, containing $X$. Distances between pairs of points in $X\setminus\{s\}$ are represented accurately by the assumption that $C$ is the tight span of $X\setminus\{s\}$. Distances from strong attachment points in $C$ to $s$ are represented accurately by the construction of $D$, and from any other point within $C$, by Lemma~\ref{lem:rect-strong-attachments}, we can find a geodesic to $s$ that stays within $C$ until it reaches a local minimum of distance to $s$ (a strong attachment point) and then follows a geodesic in $D$ to $s$. Therefore, $C\cup D$ is a hyperconvex space into which $X$ is isometrically embedded.

Finally, among hyperconvex spaces containing $X$, $C\cup D$ is minimal. For, every point in $C$ belongs to the tight span of $X\setminus\{s\}$ and therefore to the tight span of $X$. And, every point in $D$ belongs to the tight span of $P\cup\{s\}$ and therefore to the tight span of $C\cup\{s\}$, which coincides with the tight span of $(X\setminus\{s\})\cup\{s\}=X$.
\end{proof}

\subsection{Finding a good path}

By the results of the previous section, we can find the attachment points of each addition to a rectangular complex by searching for good paths. In this section we prove a uniqueness property that makes this search easier.

For a given set $S$ of boundary points of $C$, define the \emph{augmentation} of $S$ to be the union of $S$ with the set of articulation points of $C$ that lie on geodesics between two points of $S$.
We say that a block of $C$ is \emph{critical} for $S$ if it contains two or more points of the augmentation of $S$. The following lemma allows us to reduce the search for good paths to individual critical blocks.

\begin{lemma}
\label{lem:good-path-block}
Let $C$ be a rectangular complex, let $S$ be a subset of its boundary, and let $S'$ be the augmentation of $S$. Then a good path for $S$ exists if and only if the following three conditions are all met:
\begin{itemize}
\item The set of critical blocks forms a path in the tree of blocks of $C$ that was described in Lemma~\ref{lem:block-tree}.
\item Within each critical block $B$, there is a good path for $B\cap S'$
\item Each point $s$ in $S'$ that belongs to more than one critical block is an endpoint of the good path within each of the critical blocks that contain it.
\end{itemize}
If these conditions are met, then every good path for $S$ may be formed by concatenating good paths
meeting these conditions within each critical block, and every such concatenation forms a good path.
\end{lemma}

\begin{proof}
We suppose first that $P$ is a good path, and show that it can necessarily be decomposed into paths within critical blocks as described in the lemma.
Any path in $C$ can only pass through a path of the block of trees described in Lemma~\ref{lem:block-tree}.  We observe that in any block that is not critical, there is a unique articulation point $p$ that lies on any geodesic between points of the block and points of $S$; it follows that $P$ cannot enter such a block, because it is assumed to start and end at points of $S$ and would therefore have to both enter and exit through the single articulation point, a contradiction to the assumption that $P$ is a geodesic. Therefore, $P$ passes through a set of critical blocks that forms a path. This set must consist of all the critical blocks, for otherwise $P$ could not pass through all points of $S$.
Within any critical block, it must start and end either at a point of $S$ or at an articulation point; if it ends at an articulation point that does not belong to $S$, then that point lies on a geodesic between the two ends of $P$, and therefore belongs to the augmentation $S'$. Thus, within each critical block $B$, $B\cap P$ is indeed a good path for $B\cap S'$: it starts and ends at a point of $B\cap S'$, and the other conditions of being a good path are purely local. Finally, each point of $S'$ that belongs to more than one critical block must be an endpoint within each of the critical blocks that contains it, for otherwise the concatenation of the paths within these blocks would not itself be a path. Thus, if a good path $P$ exists, then the conditions of the lemma must be met, and further $P$ is formed by concatenating good paths within the critical blocks as the lemma states.

In the other direction, suppose that there exists a set of good paths within all of the critical paths, as stated by the conditions of the lemma. The union of these paths must itself be a path $P$, for the paths within each block are required to start or end at the articulation points connecting pairs of blocks, and no three blocks can share an articulation point or else the blocks would not form a path in the tree of blocks.
Concatenating geodesics within metric spaces that are connected at a single point preserves the property of being a geodesic, so $P$ must be a geodesic. $P$ must start and end at points of $S$, for all of the other starts and ends of the paths within individual blocks belong to two blocks and so do not form starts or ends of paths when the two paths within these blocks are concatenated. $P$ contains all points of $S$, for they are all contained within one of the paths in blocks. And $P$ inherits from the paths from which it was formed the other local properties of a good path: it contains only boundary points of $C$ and it does not touch two adjacent sides of a rectangle without passing through the corner. Therefore, $P$ is a good path.
\end{proof}

\begin{figure}[t]
\centering\includegraphics[width=5in]{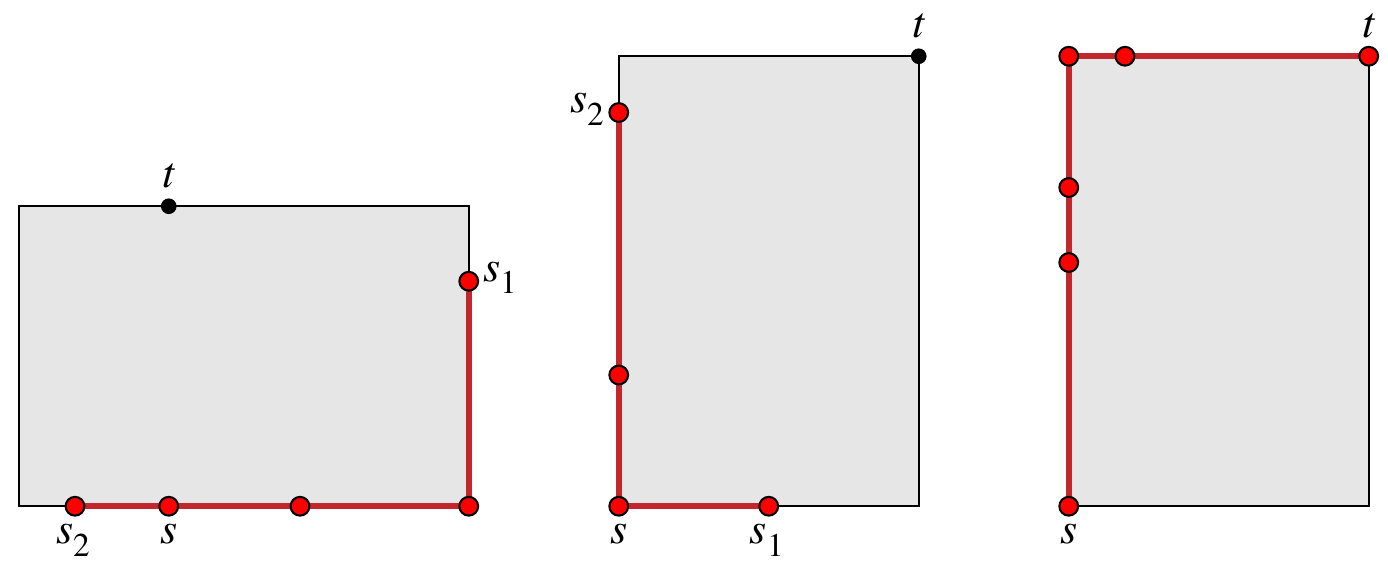}
\caption{Cases for Lemma~\ref{lem:boundary-geodesic}, showing how to find the unique good path that connects a set $S$ of points within a single block of a rectangular complex. Left: the distance from $s$ to each boundary point reaches a local minimum at $t$, and the good path has the nearest points to $t$ in either direction around the boundary as its endpoints. Center: the distance from $s$ reaches a local maximum at $t$, and $t$ does not belong to $S$; again, the good path has the nearest points to $t$ as its endpoints. Right: the distance from $s$ reaches a local maximum at a point $t$ that does belong to $S$. The good path has $s$ and $t$ as endpoints; at most one of the two boundary paths with these endpoints avoids all rectangle corners that are not in $S$.}
\label{fig:good-path-cases}
\end{figure}

\begin{lemma}
\label{lem:boundary-geodesic}
Let $S$ be any subset of boundary points of a rectangular complex $C$, given as a list of the points within each object of $C$. Then there is at most one good path for $S$ in $C$. We can test whether this path exists, and if so construct it, in time linear in the number of objects in $C$ and the number of points in~$S$.
\end{lemma}

\begin{proof}
By Lemma~\ref{lem:good-path-block} we may assume without loss of generality that $C$ has a single block: for otherwise, by induction, the statement of Lemma~\ref{lem:boundary-geodesic} holds for each critical block of $C$, and we can test whether a good path exists in $C$ by constructing the augmentation of $S$ and identifying the critical blocks, and then performing the test individually within each critical block and then verifying that the paths within each critical block obey the conditions of Lemma~\ref{lem:good-path-block}.  Recall that, when $C$ is a single block, it must be homeomorphic to a disk.

Let $s$ be any point in $S$, and consider the geodesic distances within $C$ from $s$ to each other point on the boundary of this disk. At $s$ itself, the distance from $s$ is zero, and it increases in both directions as one moves along the boundary away from $s$. If it decreases again to a local minimum $t$, then there cannot exist a geodesic from $t$ to $s$ that remains on the boundary of $C$; therefore, in this case, a good path cannot pass through $t$. The remaining portion of boundary, avoiding $t$, has the topology of an interval of a line. In this case, it is straightforward to find the necessary path: if $s_1$ and $s_2$ are the points of $S$ that are nearest to $t$ in the clockwise and counterclockwise directions around the boundary of $C$, respectively, then the good path, if it exists, must be the path from $s_1$ clockwise around the boundary to $s_2$. This case is illustrated in Figure~\ref{fig:good-path-cases}, left.

If, on the other hand, there is no local minimum of distance from $s$, then the distances from $s$ increase in both directions around the boundary of $C$ to a single local maximum $t$. There are three different subcases for this case:
\begin{itemize}
\item If $t$ does not belong to $s$, then as in the case of a local minimum $t$ cannot belong to the good path, because a geodesic from $s$ to any other point cannot pass through a local maximum of distance from $s$. Again, if $s_1$ and $s_2$ are the points of $S$ that are nearest to $t$ in the clockwise and counterclockwise directions around the boundary of $C$, respectively, then the good path, if it exists, must be the path from $s_1$ clockwise around the boundary to $s_2$ (Figure~\ref{fig:good-path-cases}, center).
\item If $t$ belongs to $s$, then both $s$ and $t$ must be endpoints of the good path, for no path that passes through one or the other of them could be a geodesic. If $S\setminus\{s,t\}$ is nonempty, then there can only be a good path if one of the two boundary paths from $s$ to $t$ contains all points in $S$ (Figure~\ref{fig:good-path-cases}, right).
\item If $S=\{s,t\}$, the previous analysis does not distinguish which of the two boundary paths from $s$ to $t$ might be a good path. But in this case, it follows from Lemma~\ref{lem:four-corners} that $C$ has at least four corners of single rectangles on its boundary. At most two of these four corners can be $s$ and $t$. The other corners do not belong to $S$, so a good path cannot go through either of them, as to do so it would have to go through two adjacent sides of the rectangles for which they are corners. Only one of the two paths from $s$ to $t$ can avoid these two other corners.
\end{itemize}
Algorithmically we may choose arbitrarily a single point $s\in S$, use Lemmas~\ref{lem:rect-distance} and~\ref{lem:edge-distance} to compute the distances of all other points on the boundary of $C$ from $s$, and use the case analysis above to determine a path $P$  along the boundary of $C$ such that, if there is a good path, it must be $P$. We may test whether $P$ is a geodesic by computing its length (the sum of the lengths of its segments) and using the distance formulas of Lemmas~\ref{lem:rect-distance} and~\ref{lem:edge-distance} to check whether its endpoints are that length apart. We may also test the other conditions for being a good path easily within the stated time bound, as they depend only on local conditions.
\end{proof}

\subsection{The incremental algorithm}

To construct the planar tight span of a metric space, we add points to the metric space one at a time, at each step representing the tight span as a rectangular complex. Specifically, we perform the following steps:

\begin{enumerate}
\item Choose arbitrarily two points of $X$ and let $C$ be a rectangular complex consisting of a single line segment connecting these two points, with length equal to the distance between them. We will maintain as an invariant that $C$ represents the tight span of the points added so far.
\item For each remaining point $s$ of $X$, in an arbitrary order:
\begin{enumerate}
\item Use Lemmas~\ref{lem:rect-distance} and~\ref{lem:edge-distance} to compute the distance from $s$ to each feature of $C$.
\item Construct the set $S$ of local minima of distance, which by Lemma~\ref{lem:rect-strong-attachments} must be the strong attachment points for $s$ in $C$.
\item If any local minimum point $p$ is calculated to have distance zero from $s$, then $s$ must already belong to the tight span of the previous points. In this case, verify that the distances from $s$ to all other sites equal the distances from $p$ to all sites, and if not abort the algorithm with an error condition. If the distances do all match, then place $s$ at point $p$ (that is, store it as one of the sites associated with the feature of $C$ that contains $p$ without otherwise changing $C$) and move on to the next point of $X$.
\item Apply Lemma~\ref{lem:boundary-geodesic} to find a good path $P$ for $S$ in $C$, if it exists.
\item If a good path does not exist, then by Lemma~\ref{lem:attachment-is-good} the tight span of the points added so far is not planar; therefore, the tight span of $X$ cannot be planar either, and we abort the algorithm with an error condition.
\item Verify that $P$ is a usable path, by identifying the local maxima of distance from $s$ along $P$ and checking the total angle in $C$ at each local maximum. If it is not usable, then as above  by Lemma~\ref{lem:attachment-is-good} the tight span is nonplanar; abort the algorithm with an error condition.
\item Apply Lemma~\ref{lem:path-extension} to construct a rectangular complex $D$ representing the tight span of $P\cup\{s\}$
\item Augment $D$ by selecting two representative sites for each bridge and four representative sites for each new rectangle, according to Lemmas~\ref{lem:rect-distance} and~\ref{lem:edge-distance}, so that distances may be computed within it. In the case of a bridge, one representative site is $s$ and the other can be any previously added site. In the case of a rectangle, three representative sites can be taken to be $s$ and the two endpoints of $P$; the fourth can be taken as any existing site that determines the distance from $s$ according to the distance formula at the fourth corner of the rectangle, a point of $P$ at which the distance from $s$ is a local maximum.
\item Replace $C$ by the union $C\cup D$, splitting any edges in $C$ or $D$ as necessary along the path $P$ where the two complexes are attached to each other. By Lemma~\ref{lem:good-is-tight} this is the tight span of the points added so far. 
\end{enumerate}
\item Return $C$ as the tight span of $X$.
\end{enumerate}

\begin{theorem}
\label{thm:incremental-correctness}
Let $(X,d)$ be a finite metric space. Then if the tight span of $(X,d)$ is planar, the algorithm above constructs a rectangular complex that correctly represents it, and if not planar then the algorithm above correctly reports that it is not planar.
\end{theorem}

\begin{proof}
The correctness of the algorithm follows from the lemmas cited within the pseudocode above: the algorithm only terminates with a failure condition when some subset of $X$ has a tight span that is guaranteed to be nonplanar, and if it terminates successfully then it has correctly found the tight span of the subset of points considered so far.
\end{proof}

\begin{theorem}
\label{thm:incremental-size-bound}
Suppose that a finite metric space $(X,d)$ has a planar tight span, and let $n=|X|$. Then the rectangular complex constructed by the algorithm described above has $O(n)$ features.
\end{theorem}

\begin{proof}
Consider the potential function $\Phi=R+3B+E$ where $R$ is the number of rectangles in the current rectangular complex, $B$ is the number of bridges, and $E$ is the number of edges that are on the boundary of a block. As we will show, after $k$ iterations of the algorithm, $\Phi\le 6k$.

Whenever a point is added, with a set $S$ of local minima of the distance function, the addition creates $|S|-1$ new rectangles and at most one new bridge, increasing $\Phi$ by $|S|+2$. In addition, if the edges at either end of $P$ were bridges, they may be split into a bridge and a boundary edge, increasing $\Phi$ by two more units total. However, at least $|S|-2$ edges of the previous rectangular complex lie entirely along the sides of newly added rectangles: this is true of every edge containing or incident to a strong attachment point, no two of which can share an edge. A bridge containing a strong attachment point becomes one or two boundary edges, and a boundary edge containing a strong attachment point becomes an interior edge of its block; these changes reduce $\Phi$ by at least $|S|-2$. Thus, in all cases, the total change to $\Phi$ per point of $X$ is at most an increase by six. It follows that, as stated above, after $k$ iterations of the algorithm $\Phi\le 6k$ and and after the algorithm completes $\Phi\le 6n$.

Each vertex of the complex is an endpoint of a bridge or a boundary edge, or an interior point that is the corner of at least two rectangles; each edge has two endpoints and each rectangle has four corners, so the number of vertices is at most $2(R+B+E)\le 2\Phi\le 12n$.
By Euler's formula, each planar subdivision with $F$ faces and $V$ vertices has exactly $F+V-2$ edges; in this case, $F=R+1$ (there is a single outer face, and each rectangle is a face) and we have bounded $V$ by $2(R+B+E)$, so there are at most $3R+2B+2E-1\le 3\Phi\le 18n$ edges.
\end{proof}

\begin{theorem}
\label{thm:incremental-complexity}
Let $(X,d)$ be a finite metric space, and let $n=|X|$. Then the algorithm described above takes time $O(n)$ per iteration, and $O(n^2)$ total time, regardless of whether it outputs the tight span of $X$ or whether it determines that the tight span is nonplanar.
\end{theorem}

\begin{proof}
Since each computation takes time linear in the size of the subdivision so far together with the number of newly added features, and we have seen that the size of the subdivision remains linear throughout the algorithm, the total time per point is linear and the total time for the whole computation is quadratic.
\end{proof}

\section{Manhattan plane embedding}
\label{sec:mpe}

\begin{figure}[t]
\centering\includegraphics[height=2in]{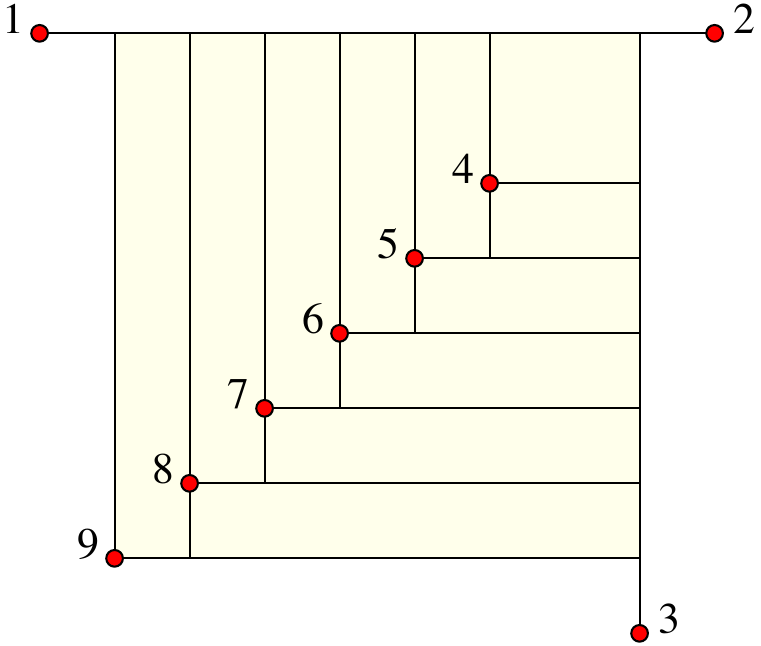}
\caption{An insertion ordering for a set of sites in the $\ell_1$ plane that leads to a rectangular complex with $2n-7$ rectangles.}
\label{fig:l1-many-rects}
\end{figure}

\begin{figure}[t]
\centering\includegraphics[scale=0.5]{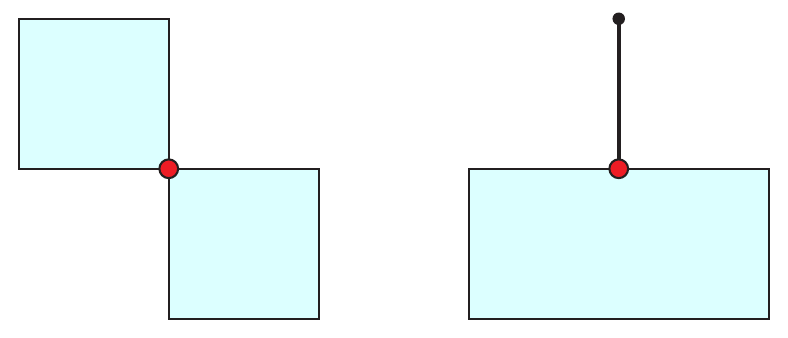}
\caption{Condition~\ref{first-condition} of Manhattan plane embeddability: If a rectangular complex is embedded into the $\ell_1$ plane, any articulation point must be the apex of two empty quadrants, either diagonally opposite each other (left) or separated by a path of vertices and bridges (right).}
\label{fig:two-empty-quads}
\end{figure}

\begin{figure}[t]
\centering\includegraphics[scale=0.5]{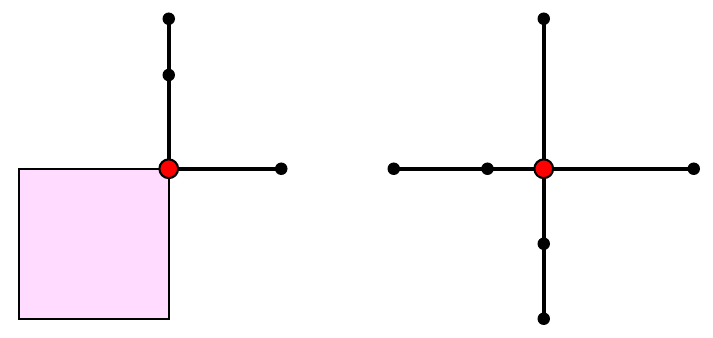}
\caption{Conditions \ref{plus-sign-cond}~and~ \ref{two-empty-quad-cond} of Manhattan plane embeddability: If a rectangular complex is embedded into the $\ell_1$ plane, for any articulation point attached to three blocks, at least two of those blocks must lead to a path of vertices and bridges (left). If it is attached to four blocks, all four must lead to paths.}
\label{fig:many-empty-quads}
\end{figure}

\begin{figure}[t]
\centering\includegraphics[scale=0.5]{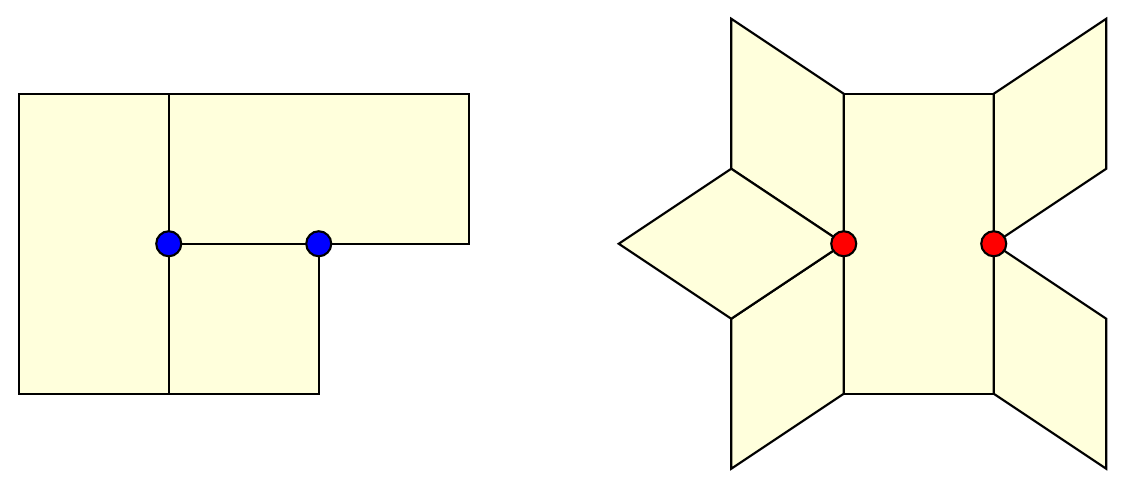}
\caption{Condition \ref{embeddable-angle-cond} of Manhattan plane embeddability: In a rectangular complex that can be embedded into the $\ell_1$ plane, the internal vertices must have a total angle of exactly $2\pi$, and the boundary vertices must have a total angle of at most $3\pi/2$, as shown by the two blue points on the left. If an internal point has an angle greater than $2\pi$, or a boundary point has an angle greater than $3\pi/2$, as do the red points on the right, then they form a cone point or cone inflection point and it is impossible to embed a neighborhood of the point into the Manhattan plane.
for any articulation point attached to three blocks, at least two of those blocks must lead to a path of vertices and bridges (left). If it is attached to four blocks, all four must lead to paths.}
\label{fig:embeddable-angles}
\end{figure}

\begin{figure}[t]
\centering\includegraphics[scale=0.5]{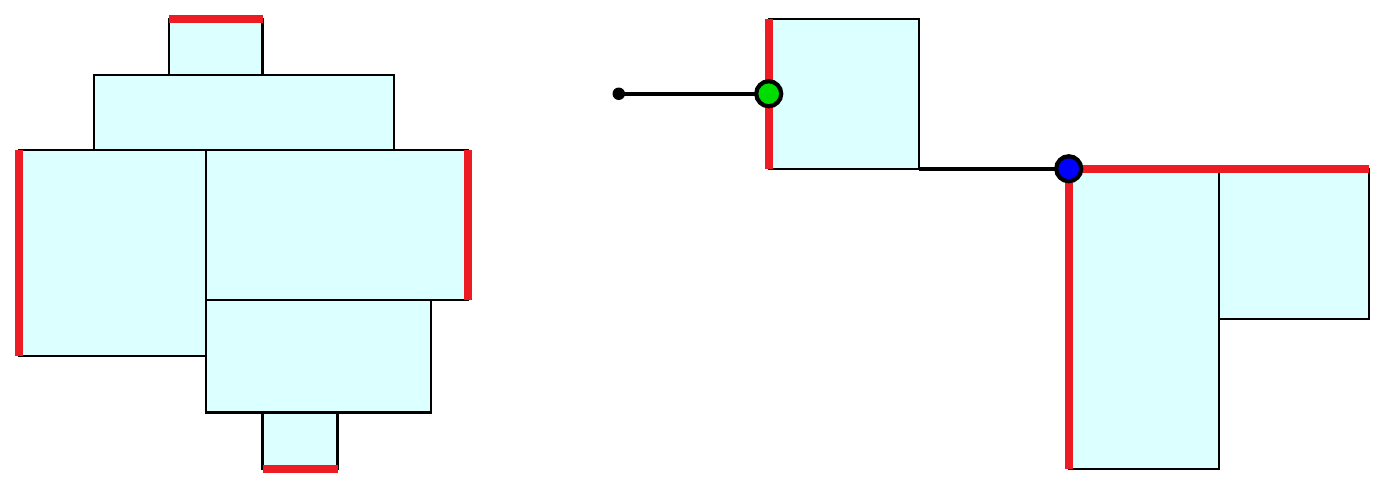}
\caption{Conditions \ref{extremal-seg-cond}--\ref{extremal-hinge-cond} of Manhattan plane embeddability: If a rectangular complex is embedded into the $\ell_1$ plane, each block must have four extremal segments (shown as the thick red segments in the left figure). Each hinge (such as the blue point on right) must belong to two extremal segments of any block it belongs to, and each other articulation point (such as the green point on the right) must belong to at least one extremal segment of any block it belongs to.}
\label{fig:extremal-segments}
\end{figure}

\begin{figure}[t]
\centering\includegraphics[scale=0.5]{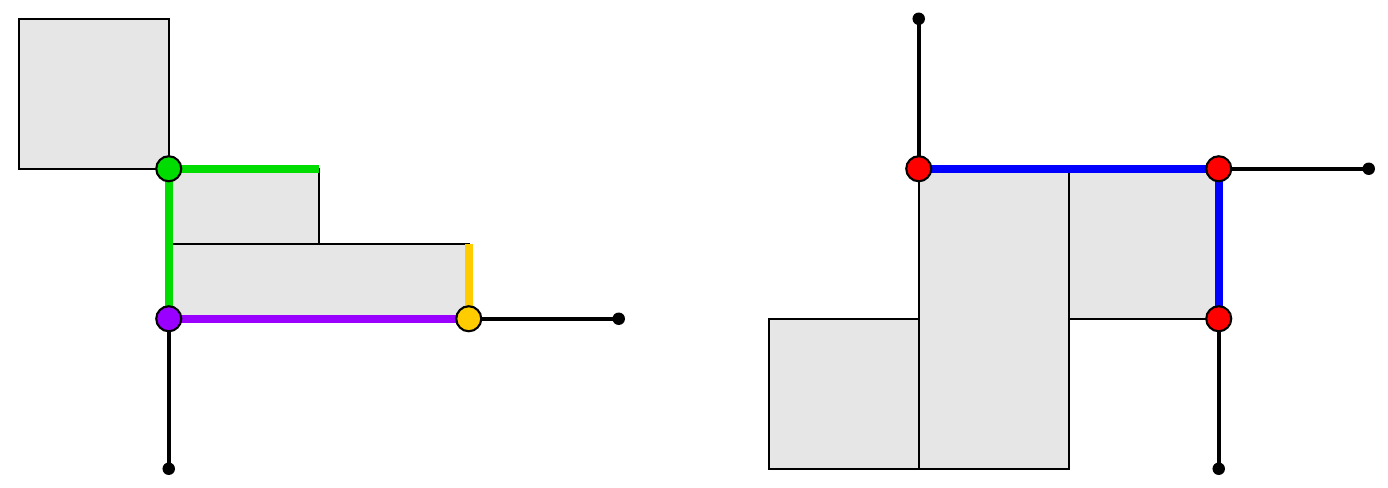}
\caption{Condition~\ref{assignment-cond} of Manhattan plane embeddability: If a rectangular complex is embedded into the $\ell_1$ plane, there must exist an assignment of two adjacent extremal segments to each hinge and one adjacent extremal segment to each non-hinge articulation point, such that each extremal segment is assigned at most once. In the left figure, the green hinge is matched with the two green extremal segments, the purple articulation point is matched to the purple segment, and the golden articulation point is matched to the golden segment.  In the rectangular complex on the right, no such assignment is possible, because the three red articulation points share only the two blue extremal segments. Therefore, no $\ell_1$ embedding is possible for the complex on the right.}
\label{fig:extremal-assignment}
\end{figure}

If a metric space can be embedded isometrically into the Manhattan plane, then so can its tight span. Therefore, we can find Manhattan plane embeddings by using the algorithm of the previous section to find a planar tight span of the given metric space, and using the structure of the tight span to determine its possible Manhattan plane embeddings.

We remark that, for metrics that can be embedded into the Manhattan plane, the constant factors in the analysis of Theorem~\ref{thm:incremental-size-bound} can be significantly tightened: the interior local minima of distance on the path connecting each new part of the tight span to the previously constructed part must each be sites, so if the algorithm creates $k$ new rectangles then it removes $k-1$ sites from the boundary of the tight span. The first rectangle cannot be created until there are four sites, and (if any rectangles are created) at least four sites remain on the boundary at the end of the algorithm. Thus, in this case, the total number of rectangles formed for a metric space with $n$ sites is at most $2n-7$. This bound is tight: there exist metric spaces formed by points in the Manhattan plane, and insertion orderings for those points, that cause our algorithm to form $2n-7$ rectangles (Figure~\ref{fig:l1-many-rects}).

We define a \emph{hinge} of a rectangular complex to be an articulation point of the complex that connects more than two blocks, or that connects exactly two blocks and has the additional property that neither of the two subcomplexes formed by cutting the complex at the articulation point has the structure of a path of bridges.

As we now discuss, there are several straightforward necessary conditions that an isometric embedding of a rectangular complex into the Manhattan plane must obey:
\begin{enumerate}
\item
\label{first-condition}
Any two blocks of the complex that meet at an articulation point must come from points that are separated geometrically by an empty open quadrant of the plane. For, if not, a geodesic within the complex from one block to the other, through the articulation point, would have a longer length than their geometric distance, contradicting the assumption that the embedding is isometric. Thus, at any articulation point, there must be at least two empty open quadrants.
 If the part of the complex on one side of the articulation point forms a path of bridges and articulation points, these two quadrants may be adjacent to each other (Figure~\ref{fig:two-empty-quads}, right); otherwise they must be diagonally opposite each other (Figure~\ref{fig:two-empty-quads}, left).
\item
\label{first-combinatorial-condition}
\label{plus-sign-cond}
Each articulation point of the complex must be incident to at most four blocks. If it is incident to exactly four blocks, the complex must have the structure of a tree with one degree-four vertex connecting four paths of bridges and articulation points (Figure~\ref{fig:many-empty-quads}, right). For, by the property above, the empty regions between each of these four blocks must occupy all four of the quadrants having the articulation point as its apex, leaving only the lines between these quadrants to be occupied by the complex. But a point set that lies entirely on one axis-aligned line will lead to a rectangular complex that is a path, and therefore the subcomplexes that are connected at the articulation point must each be paths.
\item
\label{two-empty-quad-cond}
For the same reason, if an articulation point is incident to three blocks, at least two of them must form paths of bridges and articulation points (Figure~\ref{fig:many-empty-quads}, left).
\item
\label{embeddable-angle-cond}
Within a block that is not a bridge, each boundary vertex is surrounded by rectangles forming a total angle of at most $3\pi/2$ with that vertex, and each non-boundary vertex is surrounded by rectangles forming a total angle of exactly $2\pi$ (Figure~\ref{fig:embeddable-angles}). For, otherwise the geometry in a neighborhood of that vertex would be that of a higher-order cone point or cone inflection point, not possible in the Manhattan plane.
\item
\label{extremal-seg-cond}
Within a block that is not a bridge, define a \emph{boundary segment} to be a maximal path of consecutive boundary edges and vertices with the property that the angle at each internal vertex of the path is exactly $\pi$ (Figure~\ref{fig:extremal-segments}, left). Then there must be exactly four \emph{extremal segments} for which the angles at the vertices at each end of the segment are $\pi/2$; all remaining boundary segments must have one end vertex with angle $\pi/2$ and the other with angle $3\pi/2$. For otherwise, the block could not form an orthogonally convex simple polygon in the plane.
\item
\label{extremal-articulation-cond}
Within a block that is not a bridge, each articulation point must lie on an extremal segment (as for example the green point in Figure~\ref{fig:extremal-segments}, right). For, at the other boundary points of the block, the block extends into three of the four quadrants, making it impossible to find two opposite empty quadrants for the articulation point.
\item
\label{extremal-hinge-cond}
If $B$ is a block that is not a bridge, every hinge of $B$ must coincide with the vertex where two extremal segments meet (as for example the blue point in Figure~\ref{fig:extremal-segments}, right). For again, otherwise the condition of having two opposite empty quadrants could not be met.
\item
\label{last-condition}
\label{assignment-cond}
Within any block $B$ that is not a bridge, there must exist an assignment of articulation points to extremal segments of $B$, such that each hinge is assigned to both of its adjacent extremal segments, each non-hinge articulation point is assigned to one extremal segment, and each extremal segment has at most one articulation point assigned to it (Figure~\ref{fig:extremal-assignment}). For, if a geometric embedding of the complex exists, then an assignment of this type may be determined by assigning each hinge to both incident extremal segments (Condition~\ref{extremal-hinge-cond}), and assigning each non-hinge articulation point $a$ to an incident extremal segment $s$ (Condition~\ref{extremal-articulation-cond}) such that the bridge incident to $a$ is not embedded parallel to $s$. This choice of assignment cannot assign two different articulation points to the same extremal segment because then neither of the articulation points would be surrounded by the two empty quadrants required by Condition~\ref{first-condition}.
\end{enumerate}

As we show below, Conditions \ref{first-condition}--\ref{last-condition} are necessary and sufficient for an embedding of a rectangular complex in the Manhattan plane to be isometric. Condition~\ref{first-condition} depends on the specific embedding, but Conditions \ref{first-combinatorial-condition}--\ref{last-condition} are combinatorial in nature and may be tested in linear time for a given complex without regard for its embedding, as the next lemma shows. As we also show, a rectangular complex that obeys Conditions \ref{first-combinatorial-condition}--\ref{last-condition} always has an embedding that also obeys Condition~\ref{first-condition}, and this embedding may be found in linear time.

\begin{lemma}
\label{lem:l1-test}
It is possible to test whether  a given rectangular complex obeys each of Conditions \ref{first-combinatorial-condition}--\ref{last-condition} in time linear in the number of features of the complex.
\end{lemma}

\begin{proof}
These conditions are combinatorial in nature and do not depend on a fixed embedding of the complex, so they may be tested without reference to an embedding. We may find the blocks of the complex in linear time by applying a graph connected components algorithm to the graph that has a vertex for each rectangle of the complex and an edge connecting two rectangles that share an edge.
By searching from leaves of the block tree inwards it is possible to identify each bridge of the complex that leads to a path of bridges and articulation points, and thereby determine which articulation points are hinges. With this information in hand it is straightforward to test Conditions \ref{plus-sign-cond}~and~ \ref{two-empty-quad-cond}. Additionally, Condition~\ref{embeddable-angle-cond} may be tested simply by adding the angles of the rectangles incident to each vertex and determining whether these rectangles surround the vertex in a cycle or whether the vertex is incident to some boundary edges of the complex.

Given the partition of the complex into blocks, it is possible to determine for each block the set of boundary edges, and to connect these edges by shared vertices into a cycle. Once this is done, the extremal segments of the block may be found and counted simply by examining the angles of the boundary vertices along the cycle. This allows us to test Conditions \ref{extremal-seg-cond}--\ref{extremal-hinge-cond} efficiently.

Finally, the assignment of extremal segments to hinges and articulation points (Condition~\ref{assignment-cond}) is again straightforward. If a block has more than four articulation points then it cannot be possible to assign articulation points to extremal segments as described by the condition. But otherwise there are at most four articulation points and four extremal segments to assign to them, so there are finitely many possible assignments to test and the existence of a valid assignment can be determined in constant time per block.
\end{proof}

\begin{figure}[t]
\centering\includegraphics[scale=0.5]{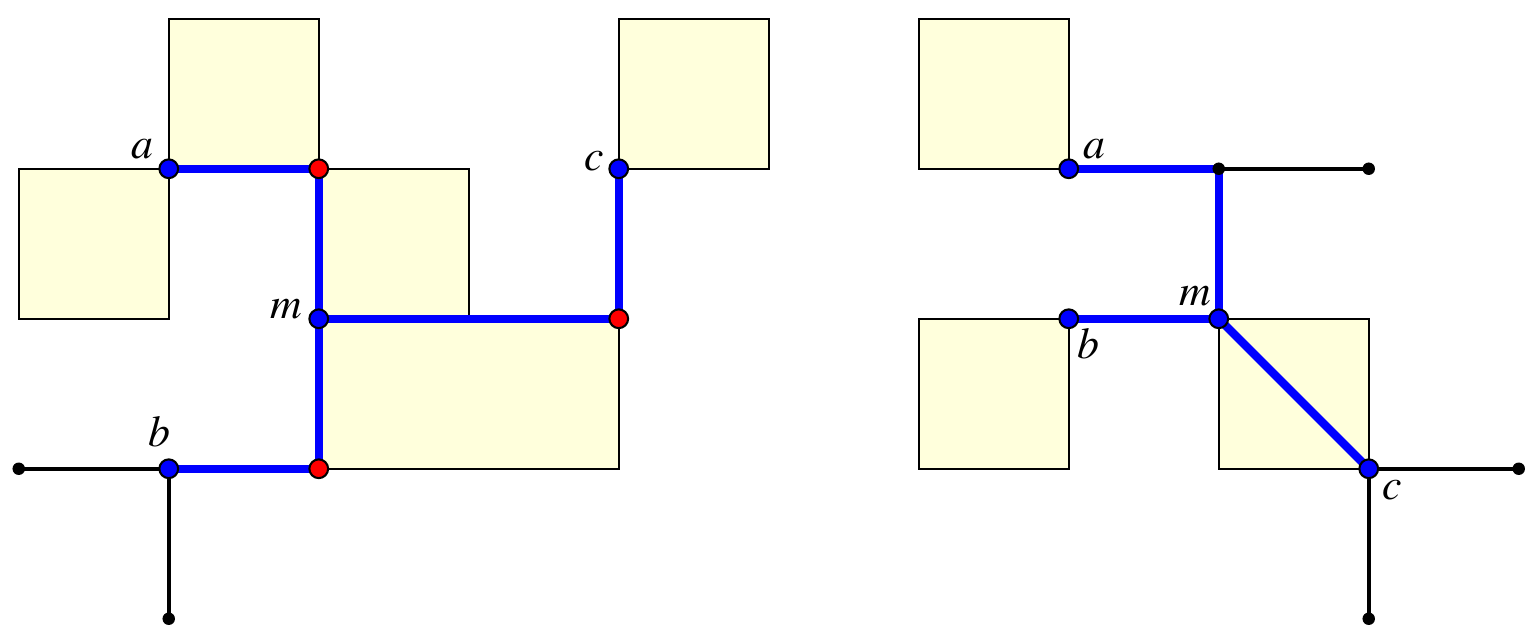}
\caption{Lemma~\ref{lem:hinge-path}: in a complex satisfying Conditions \ref{first-combinatorial-condition}--\ref{last-condition} the tight span of the hinges must be a path. For, if some three hinges do not lie on a common path, then either their median~$m$ lies within some block of the complex, and that block has too many hinges to satisfy  Condition~\ref{assignment-cond} (left), or the median is itself a hinge that violates Condition~\ref{two-empty-quad-cond} (right).}
\label{fig:clawfreehinges}
\end{figure}

\begin{lemma}
\label{lem:hinge-path}
If a rectangular complex $R$ satisfies all of Conditions \ref{first-combinatorial-condition}--\ref{last-condition}, then the tight span of the hinges in the complex forms a path.
\end{lemma}

\begin{proof}
Assume for a contradiction that $R$ satisfies all of Conditions \ref{first-combinatorial-condition}--\ref{last-condition} but that the tight span of the hinges does not form a path. Then, some triple of hinges $a$, $b$, and $c$ has a tight span $T$ that is not a path. By injectivity, $T$ may be embedded isometrically into~$R$, and its points may be identified with points of~$R$.
The tight span of three sites is either a path (with one of the three sites at a median point between the other two) or three paths connected at one central point that is not a site; by our choice of $a$, $b$, and $c$, it must be the case that $T$ is in the form of three paths connected at a distinct median point $m$. If two or more of the paths from $m$ to the three hinges $a$, $b$, and $c$ pass through a single block $B$ of $R$, then $B$ has at least three incident hinges (one on each of the paths from $m$ to $a$, $b$, and $c$, as shown in Figure~\ref{fig:clawfreehinges}, left); but then, by Condition~\ref{assignment-cond} these three hinges would have to be assigned to six different extremal segments of~$B$, contradicting Condition~\ref{extremal-seg-cond} which limits the number of extremal segments to four. If, on the other hand, $m$ is an articulation point of $R$, and the three paths from $m$ to $a$, $b$, and $c$ pass through distinct blocks (Figure~\ref{fig:clawfreehinges}, right), then $m$ is incident to three non-path subcomplexes of $R$, violating Condition~\ref{two-empty-quad-cond}. Thus in all cases we have found a contradiction to our assumptions, completing the proof.
\end{proof}

\begin{lemma}
\label{lem:embed-block}
Each block that is not a bridge in a rectangular complex $C$ that satisfies all of Conditions \ref{first-combinatorial-condition}--\ref{last-condition} has an embedding into the $\ell_1$ plane as an orthogonal polygon, with the extremal segments lying along the bounding box of the polygon. The embedding may be found in linear time and is unique up to translations, rotations by angles that are multiples of $\pi/2$, and reflections of the plane.
\end{lemma}

\begin{proof}
Choose arbitrarily one point $o$ of $C$ to be placed at the origin of the embedding and one orientation for the complex at that point. For each other point $p$ of $C$ we may determine its embedding in the Manhattan plane by tracing a curve in $C$ from $o$ to $p$ and tracing a matching curve in the plane from the origin to the image of $p$; the choice of curve is irrelevant, because any two curves from $o$ to $p$ in $C$ can be concatenated together to form a closed curve in $C$, and the shortest two curves that lead to an inconsistent placement would form a simple closed curve that does not bound a disk in $C$, contradicting the requirement that every simple closed curve in a rectangular complex bounds a disk. In particular, if $p$ and $q$ are any two points, a curve from $o$ to $p$ and then via a geodesic to $q$ is mapped in this way to a curve in the plane that includes an equal-length geodesic from $p$ to $q$, showing that this embedding is isometric. Every rectangle is thus embedded in an axis-aligned way, so the boundary of the block must form an orthogonal polygon; it must be an orthogonally convex polygon because otherwise it would not be hyperconvex, contradicting Theorem~\ref{thm:rect-complex-is-hyperconvex}. In any orthogonally convex orthogonal polygon, the extremal segments are necessarily on the sides of the bounding box.
\end{proof}

\begin{lemma}
\label{lem:l1-embed}
If a rectangular complex satisfies all of Conditions \ref{first-combinatorial-condition}--\ref{last-condition} then it can be embedded isometrically into the Manhattan plane; an embedding of this type can be found in time linear in the size of the complex.
\end{lemma}

\begin{proof}
We embed the blocks of the complex in the order given by the path of hinges from Lemma~\ref{lem:hinge-path}. At any step of this embedding let $h$ be the hinge connecting the current block $B$ to the earlier part of the path; we maintain an invariant that the part of the complex on the other side of $h$ from $B$ has already been embedded within the closed negative quadrant with respect to hinge~$h$.

As the first step of this embedding process, suppose that the given rectangular complex contains at least one hinge, and let $h$ be the first hinge on the path of hinges. We split the analysis into subcases according to what type of hinge $h$ is:
\begin{itemize}
\item If $h$ is incident to four paths of bridges and articulation points, we may place $h$ at the origin of the plane, and the four paths along the coordinate axes, completing the embedding.
\item If $h$ is  incident to two paths of bridges and articulation points and to a third block $B$, we  embed the hinge at the origin of the plane and the two paths along the two negative coordinate axes, setting up the desired invariant for $h$ and $B$.
\item If $h$ connects two blocks $B_1$ and $B_2$ that are not bridges, and block $B_1$ has no other incident hinge, then $B_1$ may be embedded in the plane by Lemma~\ref{lem:embed-block} and oriented in such a way that $h$ is the upper right vertex of its bounding box; any path connected to $B$ by a non-hinge articulation point may be embedded along a line perpendicular to the extremal segment assigned to that articulation point. In this way the desired invariant is again set up for $h$ and $B_2$.
\end{itemize}

At each intermediate step of the embedding process, we have a block $B$ incident to two hinges $h_1$ and $h_2$, where $h_1$ and the part of the complex on the other side of $h_1$ from $B$ have already been embedded in a way that satisfies the invariant. If $B$ is a bridge then we may embed $h_2$ at any point within the positive quadrant with apex at $h_1$ that is at the correct distance from $h_1$, preserving the invariant for the next step in the process.  On the other hand, if $B$ is not a bridge, an embedding of it may be given by Lemma~\ref{lem:embed-block}, in such a way that the previous hinge is the lower left corner of its bounding box and the next hinge is the upper right corner; there can be no other articulation points incident to $B$ because the two hinges together use up all four of $B$'s extremal segments. In this way the invariant that the embedded portions of the complex lie to the lower left of $h_2$ is again preserved.

When the final hinge of the path of hinges is reached, it may connect to two paths, which may be embedded on lines extending upwards and rightwards of the hinge. Alternatively, it may connect to a block with a single hinge; an embedding of this block may be given by Lemma~\ref{lem:embed-block}, placed in such a way that the hinge lies on the lower left corner of the bounding box of the block, and again any path connected to the block by a non-hinge articulation point may be embedded along a line perpendicular to the extremal segment assigned to that articulation point.

It remains to consider complexes that have no hinges. Such a complex may consist of a single path, in which case it may be embedded along a coordinate axis of the plane. Alternatively, it may have a single non-bridge block, which may be embedded by Lemma~\ref{lem:embed-block}, after which any path connected to the block by a non-hinge articulation point may be embedded along a line perpendicular to the extremal segment assigned to that articulation point.
\end{proof}

As a consequence, we have:

\begin{theorem}
\label{thm:l1-complex}
Given any rectangular complex, we may test whether it can be embedded isometrically into the Manhattan plane, and if so find an embedding, in linear time.
\end{theorem}


\begin{theorem}
\label{thm:l1-metric}
We may determine whether a finite metric space $(X,d)$ with $n$ points represents the distances between $n$ points of the Manhattan plane in time $O(n^2)$. If we have already performed this test on a metric space with $n$ points, we may add one additional point to the metric space and re-test it in time $O(n)$.
\end{theorem}

\begin{proof}
Apply Theorem~\ref{thm:incremental-complexity} to compute the tight span $C$ of $(X,d)$, and apply Theorem~\ref{thm:l1-complex} to test the embeddability of and find an embedding for $C$.
\end{proof}

\section{Discussion}
\label{sec:discuss}

We have characterized the metric spaces that may be formed as a planar tight span of a finite metric space, shown that they may be represented concisely and that distances may be computed quickly from this concise representation, and used these results to develop an efficient algorithm for constructing planar tight spans and for finding embeddings of metric spaces into the Manhattan plane.

Although it is known that a metric space can be embedded into the Manhattan plane if and only if this is true of every six-point subspace~\cite{BanChe-DCG-96} we observe that no such finite criterion exists for having a planar tight span. For, let $G$ be a graph formed from the disjoint union of a $k$-cycle (with $k\ge 4$) and a single isolated vertex, and form a metric space on the vertices of $G$ with distance one between adjacent vertices and two between nonadjacent vertices. Then $G$ has a nonplanar tight span in the form of $k$ $\ell_1$ squares of side length $1/2$ connected at an order-$k$ cone point, together with a bridge leading from that cone point to the isolated vertex. However, removing any one point from $G$ leads to a metric space with a planar tight span.

The next simplest case to investigate would seem to be the metric spaces having two-dimensional tight spans. Can these tight spans be represented concisely and constructed efficiently? As was already shown by Dress~\cite{Dre-AiM-84}, the tight span of five points is always two-dimensional, but if it is not planar then it will contain internal boundaries that are not aligned with the axes of an $\ell_1$-plane representation of its local features, so in order to handle this case it would be necessary to develop a more general representation of two-dimensional tight spans that is not based purely on rectangles; see~\cite{Kar-AC-98} and \cite{Che-AAM-00} for additional mathematical investigations of two-dimensional tight spans.

More ambitiously, it would be of interest to determine for any fixed $d$ the complexity of representing and constructing a tight span of a finite metric space that is guaranteed to be at most $d$-dimensional. The existence of polynomial time algorithms for this problem, for any fixed $d$, is not ruled out by Edmonds' NP-completeness proof~\cite{Edm-DCG-08}; rather, that proof shows that it is likely to be hard to determine whether a tight span of dimension at most three can be embedded isometrically into $\ell_\infty^3$. In recent work with Maarten L\"offler~\cite{EppLof-SoCG-11}, we found a positive answer to this question, showing that for any fixed $d$ the tight span may be represented as a set of bounded faces of a halfspace intersection with polynomial complexity, and that it may be constructed in this way in polynomial time. However, the polynomials in this result grow quadratically as a function of $d$ and it would be of interest to determine tighter bounds.

{\raggedright
\bibliographystyle{abuser}
\bibliography{injective}}

\begin{thebibliography}{10}

\bibitem{AviDez-Nw-91}
D.~Avis and M.~Deza.
\newblock {The cut cone, $L^1$ embeddability, complexity, and multicommodity
  flows}.
\newblock {\em Networks} 21(6):595{--}617, 1991,
  \href{http://dx.doi.org/10.1002/net.3230210602}%
{doi:10.1002/net.3230210602}.

\bibitem{Bad-SODA-03}
M.~B{\u{a}}doiu.
\newblock {Approximation algorithm for embedding metrics into a two-dimensional
  space}.
\newblock {\em Proc. 14th ACM-SIAM Symp. Discrete Algorithms (SODA 2003)},
  pp.~434{--}443, 2003, \url{http://dl.acm.org/citation.cfm?id=644178}.

\bibitem{BanChe-DCG-96}
H.-J. Bandelt and V.~Chepoi.
\newblock {Embedding metric spaces in the rectilinear plane: a six-point
  criterion}.
\newblock {\em Discrete Comput. Geom.} 15(1):107{--}117, 1996,
  \href{http://dx.doi.org/10.1007/BF02716581}%
{doi:10.1007/BF02716581}.

\bibitem{BanChe-Nw-98}
H.-J. Bandelt and V.~Chepoi.
\newblock {Embedding into the rectilinear grid}.
\newblock {\em Networks} 32(2):127{--}132, 1998,
  \href{http://dx.doi.org/10.1002/(SICI)1097-0037(199809)32:2$<$127::AID-NET5$%
>$3.0.CO;2-D}%
{doi:10.1002/(SICI)1097-0037(199809)32:2$<$127::AID-NET5$>$3.0.CO;2-D}.

\bibitem{BanCheEpp-SJDM-10}
H.-J. Bandelt, V.~Chepoi, and D.~Eppstein.
\newblock {Combinatorics and geometry of finite and infinite squaregraphs}.
\newblock {\em SIAM J. Discrete Math.} 24(4):1399{--}1440, 2010,
  \href{http://dx.doi.org/10.1137/090760301}%
{doi:10.1137/090760301},
  \href{http://arxiv.org/abs/0905.4537}{arXiv:0905.4537}.

\bibitem{CatCheVax-TCS-11}
N.~Catusse, V.~Chepoi, and Y.~Vax{\`e}s.
\newblock {Embedding into the rectilinear plane in optimal $O(n^2)$ time}.
\newblock {\em Theoretical Computer Science}, 2011,
  \href{http://dx.doi.org/10.1016/j.tcs.2011.01.038}%
{doi:10.1016/j.tcs.2011.01.038},
  \href{http://arxiv.org/abs/0910.1059}{arXiv:0910.1059}.

\bibitem{ChaEriWor-SoCG-08}
E.~W. Chambers, J.~Erickson, and P.~Worah.
\newblock {Testing contractibility in planar rips complexes}.
\newblock {\em Proc. 24th ACM Symp. Computational Geometry (SoCG '08)},
  pp.~251{--}259, 2008, \href{http://dx.doi.org/10.1145/1377676.1377721}%
{doi:10.1145/1377676.1377721}.

\bibitem{Che-AAM-00}
V.~Chepoi.
\newblock {Graphs of some CAT(0) complexes}.
\newblock {\em Advances in Applied Mathematics} 24(2):125{--}179, 2000,
  \href{http://dx.doi.org/10.1006/aama.1999.0677}%
{doi:10.1006/aama.1999.0677}.

\bibitem{ChrLar-SJDM-91}
M.~Chrobak and L.~L. Larmore.
\newblock {A new approach to the server problem}.
\newblock {\em SIAM J. Discrete Math.} 4:323{--}328, 1991,
  \href{http://dx.doi.org/10.1137/0404029}%
{doi:10.1137/0404029}.

\bibitem{ChrLar-Algs-94}
M.~Chrobak and L.~L. Larmore.
\newblock {Generosity helps, or an 11-competitive algorithm for three servers}.
\newblock {\em Journal of Algorithms} 16:234{--}263, 1994,
  \href{http://dx.doi.org/10.1006/jagm.1994.1011}%
{doi:10.1006/jagm.1994.1011}.

\bibitem{Dev-AoC-06}
M.~Develin.
\newblock {Dimensions of tight spans}.
\newblock {\em Annals of Combinatorics} 10(1):53{--}61, 2006,
  \href{http://dx.doi.org/10.1007/s00026-006-0273-y}%
{doi:10.1007/s00026-006-0273-y},
  \href{http://arxiv.org/abs/math.CO/0407317}{arXiv:math.CO/0407317}.

\bibitem{Dre-AiM-84}
A.~W.~M. Dress.
\newblock {Trees, tight extensions of metric spaces, and the cohomological
  dimension of certain groups}.
\newblock {\em Advances in Mathematics} 53:321{--}402, 1984,
  \href{http://dx.doi.org/10.1016/0001-8708(84)90029-X}%
{doi:10.1016/0001-8708(84)90029-X}.

\bibitem{DreHubMou-DM-01}
A.~W.~M. Dress, K.~T. Huber, and V.~Moulton.
\newblock {Metric spaces in pure and applied mathematics}.
\newblock {\em Proceedings Quadratic Forms LSU}, pp.~121{--}139, Documenta
  Mathematica, 2001,
  \url{http://www.emis.ams.org/journals/DMJDMV/lsu/dress-huber-multon.pdf}.

\bibitem{DreHubMou-AC-01}
A.~W.~M. Dress, K.~T. Huber, and V.~Moulton.
\newblock {Totally split-decomposable metrics of combinatorial dimension two}.
\newblock {\em Annals of Combinatorics} 5(1):99{--}112, 2001,
  \href{http://dx.doi.org/10.1007/PL00001294}%
{doi:10.1007/PL00001294}.

\bibitem{DreHubMou-AM-02}
A.~W.~M. Dress, K.~T. Huber, and V.~Moulton.
\newblock {An explicit computation of the injective hull of certain finite
  metric spaces in terms of their associated Buneman complex}.
\newblock {\em Advances in Mathematics} 168(1):1{--}28, 2002,
  \href{http://dx.doi.org/10.1006/aima.2001.2039}%
{doi:10.1006/aima.2001.2039},
  \url{http://www.math.uni-bielefeld.de/fsp-math/Preprints/142.pdf}.

\bibitem{Edm-DCG-08}
J.~Edmonds.
\newblock {Embedding into $l_{\infty}^2$ is easy, embedding into $l_{\infty}^3$
  is NP-complete}.
\newblock {\em Discrete Comput. Geom.} 39(4):747{--}765, 2008,
  \href{http://dx.doi.org/10.1007/s00454-008-9064-z}%
{doi:10.1007/s00454-008-9064-z}.

\bibitem{Epp-TA-09}
D.~Eppstein.
\newblock {Manhattan orbifolds}.
\newblock {\em Topology and its Applications} 157(2):494{--}507, 2006,
  \href{http://dx.doi.org/10.1016/j.topol.2009.10.008}%
{doi:10.1016/j.topol.2009.10.008},
  \href{http://arxiv.org/abs/math/0612109}{arXiv:math/0612109}.

\bibitem{EppLof-SoCG-11}
D.~Eppstein and M.~L{\"o}ffler.
\newblock {Bounds on the complexity of halfspace intersections when the bounded
  faces have small dimension}.
\newblock {\em Proc. 27th ACM Symp. Computational Geometry (SoCG '11)}, 2011,
  \href{http://dx.doi.org/10.1145/1998196.1998257}%
{doi:10.1145/1998196.1998257},
  \href{http://arxiv.org/abs/1103.2575}{arXiv:1103.2575}.

\bibitem{EppWor-WADS-09}
D.~Eppstein and K.~A. Wortman.
\newblock {Optimal embedding into star metrics}.
\newblock {\em Proc. Algorithms and Data Structures Symposium (WADS 2009)},
  pp.~290{--}301. Springer-Verlag, Lecture Notes in Computer Science 5664,
  2009, \href{http://dx.doi.org/10.1007/978-3-642-03367-4\_26}%
{doi:10.1007/978-3-642-03367-4\_26},
  \href{http://arxiv.org/abs/0905.0283}{arXiv:0905.0283}.

\bibitem{Ghr-BAMS-08}
R.~Ghrist.
\newblock {Barcodes: the persistent topology of data}.
\newblock {\em Bull. Amer. Math. Soc.} 45(1):61{--}75, 2008,
  \href{http://dx.doi.org/10.1090/S0273-0979-07-01191-3}%
{doi:10.1090/S0273-0979-07-01191-3}.

\bibitem{GhrMuh-IPSN-05}
R.~Ghrist and A.~Muhammad.
\newblock {Coverage and hole-detection in sensor networks via homology}.
\newblock {\em Proc. 4th Int. Symp. Information Processing in Sensor Networks
  (IPSN '05)}, 2005, \href{http://dx.doi.org/10.1109/IPSN.2005.1440933}%
{doi:10.1109/IPSN.2005.1440933}.

\bibitem{GolHar-SODA-05}
A.~V. Goldberg and C.~Harrelson.
\newblock {Computing the shortest path: $A^*$ search meets graph theory}.
\newblock {\em Proc. 16th ACM-SIAM Symp. Discrete Algorithms (SODA 2005)},
  pp.~156{--}165, 2005, \url{http://dl.acm.org/citation.cfm?id=1070455}.

\bibitem{HerJos-CDM-07}
S.~Herrmann and M.~Joswig.
\newblock {Bounds on the $f$-vectors of tight spans}.
\newblock {\em Contributions to Discrete Mathematics} 2(2):161{--}184, 2007,
  \href{http://arxiv.org/abs/math.MG/0605401}{arXiv:math.MG/0605401},
  \url{http://cdm.math.ca/index.php/cdm/article/viewFile/67/43}.

\bibitem{HerJosPfe-10}
S.~Herrmann, M.~Joswig, and M.~E. Pfetsch.
\newblock {Computing the bounded subcomplex of an unbounded polyhedron},
  \href{http://arxiv.org/abs/1006.2767}{arXiv:1006.2767}.
\newblock Electronic preprint, 2010.

\bibitem{Isb-CMH-64}
J.~R. Isbell.
\newblock {Six theorems about injective metric spaces}.
\newblock {\em Comment. Math. Helv.} 39:65{--}76, 1964,
  \href{http://dx.doi.org/10.1007/BF02566944}%
{doi:10.1007/BF02566944}.

\bibitem{KarOve-BIT-88}
R.~G. Karlsson and M.~H. Overmars.
\newblock {Scanline algorithms on a grid}.
\newblock {\em BIT} 28(2):227{--}241, 1988,
  \href{http://dx.doi.org/10.1007/BF01934088}%
{doi:10.1007/BF01934088}.

\bibitem{Kar-AC-98}
A.~V. Karzanov.
\newblock {Metrics with finite sets of primitive extensions}.
\newblock {\em Annals of Combinatorics} 2(3):211{--}241, 1998,
  \href{http://dx.doi.org/10.1007/BF01608533}%
{doi:10.1007/BF01608533}.

\bibitem{MalMal-DCG-92}
S.~M. Malitz and J.~I. Malitz.
\newblock {A bounded compactness theorem for $L^1$-embeddability of metric
  spaces in the plane}.
\newblock {\em Discrete Comput. Geom.} 8(1):373{--}385, 1992,
  \href{http://dx.doi.org/10.1007/BF02293054}%
{doi:10.1007/BF02293054}.

\bibitem{MonFou-TR-82}
D.~Y. Montuno and A.~Fournier.
\newblock {Finding the $x$-$y$ convex hull of a set of $x$-$y$ polygons}.
\newblock Tech. Rep. 148, University of Toronto, 1982.

\bibitem{NicLeeLia-BIT-83}
T.~M. Nicholl, D.~T. Lee, Y.~Z. Liao, and C.~K. Wong.
\newblock {On the $X$-$Y$ convex hull of a set of $X$-$Y$ polygons}.
\newblock {\em BIT} 23:456{--}471, 1983,
  \href{http://dx.doi.org/10.1007/BF01933620}%
{doi:10.1007/BF01933620}.

\bibitem{OttSoiWoo-IS-84}
T.~Ottman, E.~Soisalon-Soisinen, and D.~Wood.
\newblock {On the definition and computation of rectilinear convex hulls}.
\newblock {\em Information Sciences} 33:157{--}171, 1984,
  \href{http://dx.doi.org/10.1016/0020-0255(84)90025-2}%
{doi:10.1016/0020-0255(84)90025-2}.

\bibitem{Sch-MS-11}
R.~E. Schwartz.
\newblock {17.3 The Gauss{--}Bonnet Theorem}.
\newblock {\em Mostly Surfaces}, pp.~209{--}210. American Mathematical Society,
  Student Mathematical Library, 2011.

\bibitem{StuYu-EJC-04}
B.~Sturmfels and J.~Yu.
\newblock {Classification of six-point metrics}.
\newblock {\em Electronic Journal of Combinatorics} 11:R44, 2004,
  \href{http://arxiv.org/abs/math.MG/0403147}{arXiv:math.MG/0403147},
  \url{http://www.combinatorics.org/Volume_11/Abstracts/v11i1r44.html}.

\bibitem{Zom-CG-10}
A.~Zomorodian.
\newblock {Fast construction of the Vietoris{--}Rips complex}.
\newblock {\em Computers and Graphics} 34(3):263{--}271, June 2010,
  \href{http://dx.doi.org/10.1016/j.cag.2010.03.007}%
{doi:10.1016/j.cag.2010.03.007}.

\end{thebibliography}

\end{document}